\documentclass[sigconf]{aamas}

\usepackage{balance} 
\usepackage{booktabs} 
\usepackage{xspace}
\usepackage{mathtools}
\usepackage{amsfonts}
\usepackage{comment}
\usepackage{hyperref}
\usepackage{algorithm, float, setspace}
\usepackage[noend]{algpseudocode}
\usepackage{relsize}
\usepackage{graphicx}
\usepackage{subcaption}
\usepackage{enumitem}

\usepackage{caption}
\usepackage{xspace}
\usepackage{mathtools}
\usepackage{comment}

\usepackage{xcolor}

\usepackage[]{color-edits}
\definecolor{cerulean}{rgb}{0.10, 0.58, 0.75}
\addauthor{md}{cerulean}
\definecolor{darkred}{rgb}{0.65,0,0}
\addauthor{xw}{darkred}
\addauthor{dp}{blue}
\addauthor{dr}{magenta}

\newcommand{\algind}{\hspace{\algorithmicindent}}

\newcommand{\wo}{\backslash}
\newcommand{\inprod}{\cdot}

\newtheorem{theorem}{Theorem}
\newtheorem{lemma}{Lemma}
\newtheorem{example}{Example}
\newtheorem{definition}{Definition}

\newcommand{\Fig}[1]{Fig.~\ref{fig:#1}}
\newcommand{\Sec}[1]{Section~\ref{sec:#1}}
\newcommand{\Eq}[1]{Eq.~\eqref{eq:#1}}
\newcommand{\eq}[1]{Eq.~\eqref{eq:#1}}
\newcommand{\eqp}[1]{(Eq.~\ref{eq:#1})}
\newcommand{\eqs}[2]{Eqs.~\eqref{eq:#1} and~\eqref{eq:#2}}

\newcommand{\Algo}[1]{Algorithm~\ref{algo:#1}}

\newcommand{\Ex}[1]{Example~\ref{ex:#1}}
\newcommand{\Thm}[1]{Theorem~\ref{thm:#1}}
\newcommand{\Lem}[1]{Lemma~\ref{lem:#1}}

\newcommand{\I}{\mathcal{I}}
\newcommand{\M}{\mathcal{M}}

\newcommand{\Y}{\mathcal{Y}}

\newcommand{\Cx}{\tilde{C}}
\newcommand{\vpx}{\boldsymbol{\tilde{p}}}
\newcommand{\px}{\tilde{p}}
\newcommand{\tvtheta}{\boldsymbol{\tilde{\theta}}}
\newcommand{\tvthetaz}{\smash{\tvtheta}^{0}}

\newcommand{\pricevar}{\textit{price}}
\newcommand{\costvar}{\textit{cost}}
\newcommand{\muother}{\mu_\textit{other}}
\newcommand{\Sother}{S_\textit{other}}
\newcommand{\A}{\mathbf{A}}

\newcommand{\vb}{\mathbf{b}}

\newcommand{\va}{\mathbf{a}}
\newcommand{\one}{1}
\newcommand{\vone}{\mathbf{1}}
\newcommand{\vdelta}{\boldsymbol{\delta}}
\newcommand{\vphi}{\boldsymbol{\phi}}
\newcommand{\vtheta}{\boldsymbol{\theta}}

\newcommand{\veta}{\boldsymbol{\eta}}

\newcommand{\R}{\mathbb{R}}
\newcommand{\E}{\mathbb{E}}
\renewcommand{\Pr}{\mathbb{P}}
\newcommand{\vzero}{\boldsymbol{0}}
\newcommand{\vmu}{{\boldsymbol{\mu}}}

\newcommand{\sveta}{\boldsymbol{\eta^\star}}
\newcommand{\vp}{\boldsymbol{p}}

\newcommand{\Z}{\mathcal{Z}}

\newcommand{\lt}[1]{\textit{left}(#1)}
\newcommand{\rt}[1]{\textit{right}(#1)}
\newcommand{\pt}[1]{\textit{par}(#1)}
\newcommand{\sib}[1]{\textit{sib}(#1)}

\newcommand{\rootnode}{\textit{root}}

\newcommand{\tomega}{\nu}
\newcommand{\tree}{T}
\newcommand{\order}{\mathcal{O}}
\newcommand{\price}{\textbf{\upshape price}}
\newcommand{\buy}{\textbf{\upshape buy}}
\newcommand{\cost}{\textbf{\upshape cost}}
\newcommand{\vals}{\textit{vals}}
\newcommand{\nbuy}{n_{\textit{buy}}}
\newcommand{\nvals}{n_{\textit{vals}}}
\newcommand{\level}{\textit{level}}
\newcommand{\ttheta}{\tilde{\theta}}
\newcommand{\ly}{y_\textit{l}}
\newcommand{\ry}{y_\textit{r}}

\newcommand{\abs}[1]{\lvert#1\rvert}
\newcommand{\card}[1]{\lvert#1\rvert}
\newcommand{\Bracks}[1]{\left[#1\right]}
\newcommand{\bigBracks}[1]{\bigl[#1\bigr]}

\newcommand{\parens}[1]{(#1)}
\newcommand{\Parens}[1]{\left(#1\right)}
\newcommand{\bigParens}[1]{\bigl(#1\bigr)}
\newcommand{\BigParens}[1]{\Bigl(#1\Bigr)}
\newcommand{\biggParens}[1]{\biggl(#1\biggr)}
\newcommand{\set}[1]{\{#1\}}
\newcommand{\braces}[1]{\{#1\}}

\renewcommand{\prec}{\textit{prec}}

\newcommand{\snp}{S\&P\,500\xspace}

\algdef{SE}[SUBALG]{Indent}{EndIndent}{}{\algorithmicend\ }%
\algtext*{Indent}
\algtext*{EndIndent}

\setcopyright{ifaamas}
\acmConference[AAMAS '21]{Proc.\@ of the 20th International Conference on Autonomous Agents and Multiagent Systems (AAMAS 2021)}{May 3--7, 2021}{Online}{U.~Endriss, A.~Now\'{e}, F.~Dignum, A.~Lomuscio (eds.)}
\copyrightyear{2021}
\acmYear{2021}
\acmDOI{}
\acmPrice{}
\acmISBN{}

\acmSubmissionID{612}

\title{Log-time Prediction Markets for Interval Securities}
\author{Miroslav Dud\'ik$^*$}
\affiliation{
  \institution{Microsoft Research, New York, NY}}
\email{mdudik@microsoft.com}
\thanks{$^*$Authors contribute equally.}

\author{Xintong Wang$^{*}$}
\affiliation{
  \institution{University of Michigan, Ann Arbor, MI}}
\email{xintongw@umich.edu}

\author{David M. Pennock}
\affiliation{
  \institution{Rutgers University, New Brunswick, NJ}}
\email{dpennock@dimacs.rutgers.edu }

\author{David M. Rothschild}
\affiliation{
	\institution{Microsoft Research, New York, NY}}
\email{davidmr@microsoft.com}

\begin{abstract}
We design a prediction market to recover a complete and fully general probability distribution over a random variable.
Traders buy and sell \textit{interval securities} that pay \$1 if the outcome falls into an interval and \$0 otherwise.
Our market takes the form of a central \textit{automated market maker} and allows traders to express interval endpoints of arbitrary precision.
We present two designs in both of which market operations take time logarithmic in the number of intervals (that traders distinguish), providing the first computationally efficient market for a continuous variable.
Our first design replicates the popular \textit{logarithmic market scoring rule} (LMSR), but operates exponentially faster than a standard LMSR by exploiting its modularity properties to construct a balanced binary tree and decompose computations along the tree nodes.
The second design consists of two or more parallel LMSR market makers that mediate submarkets of increasingly fine-grained outcome partitions.
This design remains computationally efficient for all operations, including arbitrage removal across submarkets.
It adds two additional benefits for the market designer:
(1)~the ability to express utility for information at various resolutions by assigning different liquidity values, and
(2)~the ability to guarantee a true constant bounded loss by appropriately decreasing the liquidity in each submarket.
\end{abstract}


\newcommand{\BibTeX}{\rm B\kern-.05em{\sc i\kern-.025em b}\kern-.08em\TeX}

\begin{document}


\pagestyle{fancy}
\fancyhead{}

\maketitle

%
\section{Introduction}
\label{sec:intro}
Consider a one-dimensional random variable, such as the opening value of the \snp index on December 17, 2021.
We design a market for trading \emph{interval securities} corresponding to predictions that the outcome will fall into some specified interval, say between 2957.60 and 3804.59, implemented as binary contracts that pay out \$1 if the outcome falls in the interval and \$0 otherwise.
We are interested in designing \emph{automated market makers} to facilitate a fully \emph{expressive} market computationally efficiently.
Traders can select custom interval endpoints of arbitrary precision corresponding to a continuous outcome space, whereas the market maker will always offer to buy or sell any interval security at some price.

A form of interval security called the \textit{condor spread} is common in \textit{financial options} markets, with significant volume of trade.
Each condor spread involves trading four different options,%
\footnote{A \textit{call option} written on an underlying stock with \textit{strike price} $K$ and expiration date $T$ pays $\max\{S-K,0\}$, where $S$ is the opening price of the stock on date $T$.
	For example, 25 shares of
	``\$1 iff [2650,2775]''
	${}\approx{}$
	$\max\{S-2650,0\}-\max\{S-2675,0\}-\max\{S-2750,0\}+\max\{S-2775,0\}$.\looseness=-1}
and financial options offered by the market may only support a limited subset of approximate intervals.
As of this writing, \snp options expiring on December 17, 2021, distinguish 56 strike prices, allowing the purchase of around 1500 distinct intervals of minimum width 25.
Moreover, as each strike price trades independently despite the logical constraints on their relative values, it will require time linear in the number of offered strike prices to remove arbitrage.

Outside traditional financial markets, the \emph{logarithmic market scoring rule} (LMSR) market maker~\cite{Hanson03,Hanson07} has been used to elicit information through the trade of interval securities. 
The Gates Hillman Prediction Market at Carnegie Mellon University operated LMSR on 365 outcomes, representing 365 days of one year, to forecast the opening time of the new computer science building~\cite{Othman10}.
Traders could bet on different intervals by choosing a start and an end date. 
A similar market\footnote{\url{www.cs.utexas.edu/news/2012/research-corner-gates-building-prediction-market}}
was later launched at the University of Texas at Austin, using a liquidity-sensitive variation of LMSR~\cite{Othman13}.
Moreover, LMSR has been deployed to predict product-sales levels~\cite{Chen2002}, instructor ratings~\cite{Chakraborty13}, and political events~\cite{hanson1999}.

LMSR has two limitations that prevent its scaling to markets with a continuous outcome space.
First, LMSR's worst-case loss can grow unbounded if traders select intervals with prior probability approaching zero \cite{GaoChenPennock09}.
Second, standard implementations of LMSR operations run in time linear in the number of outcomes or distinct future values traders define---in our case, arbitrarily many.
The constant-log-utility and other barrier-function-based market makers~\cite{ChenPe07,Othman12} achieve constant bounded loss, but still suffer the second limitation regarding computational intractability.
Thus, previous markets allow only a relatively small set of predetermined intervals and run in time linear in the number of supported outcomes, limiting the ability to aggregate high-precision trades and elicit the full distribution of a continuous random variable.

In this paper, we propose two automated market makers that perform exponentially faster than the standard LMSR and previous designs.
Market operations (i.e., $\price$, $\cost$, and $\buy$) can be executed in time \emph{logarithmic} in the number of distinct intervals traded, or linear in the number of bits describing the outcome space. 
Our first market maker calculates LMSR exactly, but employs a balanced binary tree to implement interval queries and trades. 
We show that the normalization constant of LMSR---a key quantity in its price and cost function---can be calculated recursively via local computations on the balanced tree.
Our work here contributes to the rich literature that aims to overcome the worst-case \#P-hardness of LMSR pricing~\cite{ChenEtAl08} by exploiting the outcome space structure and limiting  expressivity~\cite{Chen:07,Guo:09,Chen:08b,XiaPe11,LaskeyEtAl18}.


Our second market maker works by maintaining parallel LMSR submarkets that adopt different liquidity parameters and offer interval securities at various resolutions. 
We show that liquidity parameters can be chosen to guarantee a \emph{constant} bounded loss independent of market precision and prices can be kept coherent efficiently by removing arbitrages across submarkets.
We demonstrate through agent-based simulation that our second design enjoys more flexible liquidity choices to facilitate the information-gathering objective: it can get close to the ``best of both worlds'' displayed by coarse and fine LMSR markets, with prices converging fast at both resolutions regardless of the traders' information structure.

The two proposed designs, to our knowledge, are the first to simultaneously achieve expressiveness and computational efficiency.
As both market makers facilitate trading intervals at arbitrary precision, they can elicit any probability distribution over a continuous random variable that can be practically encoded by a machine.
We use the \snp index value as a running example,
but our framework is generic and can handle any one-dimensional continuous variable, for example,
the landfall point of a hurricane along a coastline
or the number of tickets sold in the first week of a movie release.\looseness=-1

\section{Formal Setting}
\label{sec:setting}
We first review cost-function-based market making~\cite{AbernethyChVa11,ChenPe07}, and then introduce interval markets.
\subsection{Cost-Function-Based Market Making}
\label{sec:cost}
Let $\Omega$ denote a finite set of \emph{outcomes}, corresponding to mutually exclusive and exhaustive states of the world.
We are interested in eliciting expectations of binary random variables $\phi_i \colon \Omega\to\{0,1\}$, indexed by $i\in\I$, which model the occurrence of various events, such as ``\emph{\snp will open between 2957.60 and 3804.59 on December 17, 2021}''.
Each variable $\phi_i$ is associated with a \emph{security} that pays out $\phi_i(\omega)$ when the outcome $\omega\in\Omega$ occurs, and thus $\phi_i$ is also called the \emph{payoff function}.
Binary securities pay out \$1 if the specified event occurs and \$0 otherwise.
The vector $(\phi_i)_{i\in\I}$ is denoted~$\vphi$.
Traders trade \emph{bundles} $\vdelta \in \R^{\card{\I}}$ of security with a central market maker, where positive entries in $\vdelta$ correspond to purchases and negative entries to short sales.
A trader holding a bundle $\vdelta$ receives a payoff of $\vdelta\inprod\vphi(\omega)$, when $\omega$ occurs.

Following~\cite{AbernethyChVa11} and~\cite{ChenPe07}, we assume that the market maker determines security prices using a convex and differentiable potential function $C \colon \R^{\card{\I}}\to\R$, called a \emph{cost function}.
The state of the market is specified by a vector $\vtheta\in\R^{\card{\I}}$, listing the number of shares of each security \emph{sold} by the market maker so far.
A trader who wants to buy a bundle $\vdelta$ in the market state $\vtheta$ must pay $C(\vtheta+\vdelta)-C(\vtheta)$ to the market maker, after which the new state becomes $\vtheta+\vdelta$.

The vector of instantaneous prices in the corresponding state $\vtheta$ is $\vp(\vtheta)\coloneqq\nabla C(\vtheta)$.
Its entries can be interpreted as the market's collective estimates of $\E[\phi_i]$: a trader can make an expected profit by buying (at least a small amount of) the security $i$ if she believes that $\E[\phi_i]$ is larger than the instantaneous price $p_i(\vtheta)=\partial C(\vtheta)/\partial\theta_i$, and by selling
if she believes the opposite.
Therefore, risk neutral traders with sufficient budgets maximize their expected profits by moving the price vector to match their expectation of~$\vphi$.
Any expected payoff must lie in the convex hull of the set $\{\vphi(\omega)\}_{\omega \in \Omega}$, which we denote $\M$ and call a \emph{coherent price space} with its elements referred to as \emph{coherent price vectors}.

We assume that the cost function satisfies two standard properties: \emph{no arbitrage} and \emph{bounded loss}.
The \emph{no-arbitrage} property requires that as long as all outcomes $\omega$ are possible, there be no market transaction with a guaranteed profit for a trader.
In this paper, we use the fact that $C$ is arbitrage-free if and only if it yields price vectors $\vp(\vtheta)$ that are always coherent \cite{AbernethyChVa11}.
The \emph{bounded-loss} property is defined in terms of the worst-case loss of a market maker,
$\sup_{\vtheta\in\R^{\card{\I}}}\sup_{\omega\in\Omega}\bigBracks{\vtheta\inprod\vphi(\omega)-C(\vtheta)+C(\vzero)}$,
meaning the largest difference, across all possible trading sequences and outcomes, between the amount that the market maker has to pay the traders (once the outcome is realized) and the amount that the market maker has collected (when securities were traded).
The property requires that this worst-case loss be a priori bounded by a constant.


\subsection{Complete Markets and LMSR}
\label{sec:lmsr}
In a complete market, we have $\I=\Omega$.
Securities are indicators of individual outcomes, $\phi_i(\omega)=\one\set{\omega=i}$, where $\one\set{\cdot}$ denotes the binary indicator. 
We denote each market security as $\phi_\omega$.
A risk-neutral trader is incentivized to move the price of each security $\phi_\omega$ to her estimate of $\E[\phi_\omega]=\Pr[\omega]$, which is her subjective probability of $\omega$ occurring.
Thus, traders can express arbitrary probability distributions over $\Omega$.
We consider variants of LMSR market maker~\cite{Hanson03} for a complete market,
described by cost function and prices
\begin{equation}
\label{eq:LMSR}
\!C(\vtheta)=b\log\Parens{\sum_{\omega\in\Omega} e^{\theta_\omega / b}}, \quad
p_\omega(\vtheta)=\frac{\partial C(\vtheta)}{\partial\theta_\omega} = \frac{e^{\theta_\omega/b}}{\sum_{\tomega\in\Omega} e^{\theta_{\tomega}/b}},\!
\end{equation}
where $b$ is the liquidity parameter, controlling how fast the price moves in response to trading and limiting the worst-case loss of the market maker to $b\log\,\card{\Omega}$~\cite{Hanson03}.


The securities in a complete market can be used to express bets on any event $E$.
%
Specifically, one share of a security for the event $E$ can be represented by the indicator bundle $\vone_E\in\R^\Omega$ with entries $1_{E,\omega}=\one\set{\omega\in E}$.
We refer to this bundle as the \emph{bundle security for event $E$}.
%
The immediate price of the bundle $\vone_E$ in the state $\vtheta$ is
\begin{equation}
\label{eq:pE}
p_E(\vtheta)\coloneqq
\vone_E\inprod\vp(\vtheta)
=\sum_{\omega\in E} p_\omega(\vtheta)
=\frac{
	\sum_{\omega\in E} e^{\theta_\omega/b}
}{
	\sum_{\tomega\in\Omega} e^{\theta_{\tomega}/b}
}
.
\end{equation}
The cost of buying the bundle $s\vone_E$, or sometimes referred to as ``the cost of $s$ shares of $\vone_E$'', can be written
as a function of $p_E(\vtheta)$ and $s$:
\begin{align}
&C(\vtheta+s\vone_E)-C(\vtheta) \label{eq:costE}\\[1em]
&\;\;{}=
\smash[t]{
	b\log\Parens{\sum_{\omega\not\in E}e^{\theta_\omega/b} + \sum_{\omega\in E}e^{(\theta_\omega+s)/b}}
	-b\log\Parens{\sum_{\omega\in\Omega} e^{\theta_\omega/b}}
} \notag\\
&\;\;{}=
b\log\Parens{p_{E^c}(\vtheta) + e^{s/b}p_E(\vtheta)}
=
b\log\Parens{1-p_{E}(\vtheta) + e^{s/b}p_E(\vtheta)}. \notag
\end{align}
Above, we write $E^c$ for the complementary event $E^c=\Omega\wo E$, and use the fact ${p_{E}(\vtheta)+p_{E^c}(\vtheta)=1}$, which follows from \eq{pE}.

\subsection{Interval Securities over $[0,1)$}
We consider betting on outcomes within an interval $[0,1)$.
Our approach generalizes to outcomes that are in any $[\alpha,\beta)\subseteq[-\infty,\infty)$ by applying any increasing transformation $F:[\alpha,\beta)\to[0,1)$. 
We assume that the outcome $\omega$ is specified with $K$ bits, meaning that there are $N=2^K$ outcomes with $\Omega=\set{j/N:\:j\in\set{0,1,\dotsc,N-1}}$.
At the end of Sections~\ref{sec:log-time mm} and~\ref{sec:const-loss}, we discuss how the assumption of pre-specified bit precision can be removed.
\begin{example}[Complete market for \snp]
	We construct a complete market for the \snp opening price on December 17, 2021, by setting $N=2^{19}=\text{524,288}$.
	The resulting complete market is $\I=\set{0,\,0.01,\,\dotsc,\,5242.86,\,5242.87}$, where we cap prices at \$5242.87 (i.e., larger prices are treated as \$5242.87).
	The transformed outcome is then $\omega=\omega'/N$, where $\omega'$ is the $\snp$ price in cents.
\end{example}

In the outcome space $\Omega$, we would like to enable price and cost queries as well as buying and selling of \emph{bundle securities} for the interval events $I=[\alpha,\beta)$ for any $\alpha,\beta\in\Omega\cup\set{1}$.
For cost-based markets, sell transactions are equivalent to buying a negative amount of shares, so we design algorithms for three operations: $\price(I)$, $\cost(I,s)$, and $\buy(I,s)$, where $I$ is the interval event and $s$ the number of shares.
A naive implementation of $\price$ and $\cost$ following \eqs{pE}{costE} would be \emph{linear} in $N$.
In this paper, we propose to implement these operations in time that is \emph{logarithmic} in $N$. 



\section{A Log-time LMSR Market Maker} 
\label{sec:log-time mm}
We design a data structure, referred to as an \emph{LMSR tree}, which resembles an \emph{interval tree}~\cite[Section 15.3]{CLR99}, but includes additional annotations to support LMSR calculations.
We first define the LMSR tree, and show that it can facilitate market operations in time logarithmic in the number of distinct intervals that traders define. 

\subsection{An LMSR Tree for $[0, 1)$}
We represent an LMSR tree $T$ with a \emph{full binary tree}, where each node $z$ has either no children (when $z$ is a leaf) or exactly two children, denoted $\lt{z}$ and $\rt{z}$ (when $z$ is an inner node).
The root is denoted $\rootnode$ and the parent of any non-root node $\pt{z}$.
\begin{definition}[LMSR Tree]
	\label{def:lmsr_tree}
	An \emph{LMSR tree} is a full binary tree, where each node $z$ is annotated with an interval $I_z=[\alpha_z,\beta_z)$ with $\alpha_z,\beta_z\in\Omega\cup\set{1}$, a height $h_z \ge 0$, a quantity $s_z \in\R$ that records the number of sold bundle securities associated with $I_z$, and a partial normalization constant $S_z \ge 0$ (defined below in Eq.~\ref{eq:Sz}).

An LMSR tree is required to satisfy:
	\begin{itemize}[leftmargin=*]
		\item \emph{Binary-search property}: $I_{\rootnode}=[0,1)$, and for inner node $z$,
		\[\alpha_z=\alpha_{\lt{z}}\;<\;\beta_{\lt{z}}=\alpha_{\rt{z}}\;<\;\beta_{\rt{z}}=\beta_z.\]
		\item \emph{Height balance}: $h_z=0$ for leaves, and for inner node $z$,
		\[h_z=1+\max\set{h_{\lt{z}},h_{\rt{z}}}, \quad \abs{h_{\lt{z}}-h_{\rt{z}}}\le 1.\]
		\item \emph{Partial-normalization correctness}:
		$S_z= e^{s_z/b}\cdot (\beta_z-\alpha_z)$ for leaves, and for inner node $z$,
		\[S_z=e^{s_z/b}\cdot \Parens{S_{\lt{z}}+S_{\rt{z}}}.\]
	\end{itemize}
\end{definition}
%
The \emph{binary-search property} helps to find the unique leaf that contains any $\omega\in\Omega$ by descending from $\rootnode$ and choosing left or right in each node based on whether $\omega<\beta_{\lt{z}}$ or $\omega\ge\beta_{\lt{z}}$.
The \emph{height-balance property} ensures that the path length from
root to any leaf is at most $\order(\log n)$, where $n$ is the number of leaves of the tree~\cite{Knuth}.
We adopt an \emph{AVL tree}~\cite{AVL62} at the basis of our LMSR tree, but other balanced binary-search trees (e.g., red-black trees or splay trees) could also be used.

To facilitate LMSR computations, we maintain a scalar quantity $s_z\in\R$ for each node $z$, which records the number of \emph{bundle securities} associated with $I_z$ sold by the market maker.
Therefore, the market state and its components for each individual outcome $\omega$ represented by the LMSR tree $\tree$ are%
\footnote{We write $\omega\in z$ to mean $\omega\in I_z$ and $z'\subseteq z$ to mean $I_{z'}\subseteq I_{z}$. Thus, $z'\subseteq z$ means that $z'$ is a descendant of $z$ in $T$, and $z'\subset z$ means that $z'$ is a strict descendant of $z$.}
\begin{equation}
\label{eq:thetaT}
\vtheta(\tree)=\smash[t]{\sum_{z\in\tree} s_z\vone_{I_z}};
\quad \theta_{\omega}(T)=\sum_{z\in T} s_z 1_{I_z,\omega}
=\sum_{z\ni\omega} s_z.
\end{equation}
%
%
The normalization constant in the LMSR price \eqp{pE} is then
\begin{equation}
\label{eq:Sroot}
\sum_{\omega\in\Omega} e^{\theta_\omega/b}
=\sum_{\omega\in\Omega} e^{\sum_{z\ni\omega} s_z/b}
=\sum_{\omega\in\Omega} \prod_{z\ni\omega} e^{s_z/b}.
\end{equation}
%
We decompose the computation of the above normalization constant along the nodes of an LMSR tree, by defining a \emph{partial normalization constant} $S_z$ in each node:%
%
\begin{equation}
\label{eq:Sz}
S_z\coloneqq
\frac1N
\adjustlimits\sum_{\omega\in z}\prod_{\;\;z':\:z\supseteq z'\ni\omega\;\;}
\!\!e^{s_{z'}/b}.
\end{equation}
Thus, we have $\sum_{\omega\in\Omega} e^{\theta_\omega/b}=NS_\rootnode$ and obtain the following recursive relationship, which we refer to as \emph{partial-normalization correctness} and is at the core of implementing $\price$ and $\buy$:\looseness=-1
\begin{equation}
\label{eq:Sz:rec}
S_z
=
\begin{cases}
e^{s_z/b}\cdot (\beta_z-\alpha_z)
&\text{if $z$ is a leaf,}
\\
e^{s_z/b}\cdot \Parens{S_{\lt{z}}+S_{\rt{z}}}
&\text{otherwise.}
\end{cases}
\end{equation}

Based on the LMSR tree construction, we implement the following operations for any interval $I=[\alpha,\beta)$:
\begin{itemize}[leftmargin=*]
	\item $\price(I,\tree)$ returns the price of bundle security for $I$;
	\item $\cost(I,s,\tree)$ returns the cost of $s$ shares of bundle security for $I$;\looseness=-1
	\item $\buy(I,s,\tree)$ updates $\tree$ to reflect the purchase of $s$ shares of bundle security for $I$.
\end{itemize}
For $\cost$, it suffices to implement $\price$ and use \eq{costE}.
Since the price of $[\alpha,\beta)$ satisfies $p_{[\alpha,\beta)}(\vtheta)=p_{[\alpha,1)}(\vtheta)-p_{[\beta,1)}(\vtheta)$, it
suffices to implement $\price$ for intervals of the form $[\alpha,1)$.
Similarly, buying $s$ shares of $[\alpha,\beta)$ is equivalent to first buying $s$ shares of $[\alpha,1)$ and then buying $(-s)$ shares of $[\beta,1)$, as the market ends up in the same state $\vtheta+s\vone_{[\alpha,\beta)}$.
We implement $\price$ and $\buy$ for one-sided intervals $I=[\alpha,1)$, and the remaining operations will follow.

\subsection{Price Queries}
We consider price queries for $I=[\alpha,1)$.
Let $\vals(T)=\set{\alpha_z:\:z\in T}$ denote the set of distinct left endpoints in the tree nodes.
We start by assuming that $\alpha\in\vals(T)$, and later relax this assumption.
We proceed to calculate $p_I(\vtheta)$ in two steps.
\emph{First}, we construct a set of nodes $\Z$ whose associated intervals $I_z$ are disjoint and cover $I$.
To achieve this, we conduct a binary search for $\alpha$, putting in $\Z$ all of the right children of the visited nodes that have $\alpha_z>\alpha$, as well as the final node with $\alpha_z=\alpha$.
Thanks to the height balance, the cardinality of $\Z$ is $\order(\log n)$, where $n$ is the
number of leaves of $T$.
The resulting set $\Z$ satisfies $p_I(\vtheta)=\sum_{z\in\Z} p_{I_z}(\vtheta)$.

\emph{Second}, we determine $p_{I_z}(\vtheta)$ for each node $z\in\Z$.
Starting from the LMSR price in \eq{pE}, we take advantage of the defined partial normalization constants $S_z$ to calculate $p_{I_z}(\vtheta)$:
\begin{align}
\label{eq:pIz:expand}
p_{I_z}(\vtheta)
&=\smash[t]{\frac{1}{NS_{\rootnode}}\sum_{\omega\in z} e^{\theta_\omega/b}}
=\frac{1}{S_{\rootnode}}\cdot\frac1N
\adjustlimits\sum_{\omega\in z}\prod_{z'\ni\omega}e^{s_{z'}/b}\\
\label{eq:pIz:decompose}
&=\frac{1}{S_{\rootnode}}\cdot\frac1N\sum_{\omega\in z}\Bracks{
	\Parens{\prod_{z':\:z\supseteq z'\ni\omega}
		e^{s_{z'}/b}}
	\Parens{\prod_{z'\supset z} e^{s_{z'}/b}}}\\
\label{eq:pIz}
&=\frac{S_z}{S_{\rootnode}}
\,\,\underbrace{\!\!\Parens{\prod_{z'\supset z} e^{s_{z'}/b}}\!\!}_{P_z}\,\,.
\end{align}
In \eq{pIz:expand}, we use that $NS_\rootnode=\sum_{\omega\in\Omega} e^{\theta_\omega/b}$
and then expand $\theta_\omega$ using \eq{thetaT}.
In \eq{pIz:decompose}, we use the fact that any node $z'$ with a non-empty intersection with $z$ (i.e., $I_z\cap I_{z'}\ne\emptyset$) must be either a descendant or an ancestor of $z$ as a direct consequence of the binary-search property.
The product $P_z$ in \Eq{pIz} iterates over $z'$ on the path from root to $z$, and thus can be calculated along the binary-search path.\looseness=-1
%

We now handle the case when $\alpha\not\in\vals(T)$.
After the leaf $z$ on the search path is reached, we have $\alpha_z<\alpha<\beta_z$.
Instead of expanding the tree, we conceptually create two children of $z$: $z'$ and $z''$ with $I_{z'}=[\alpha_z,\alpha)$ and $I_{z''}=[\alpha,\beta_z)$, and add $z''$ in $\Z$.
Since $\theta_{\omega}$ is constant across $\omega\in I_z$, we obtain $p_{I_{z''}}(\vtheta)=\frac{\beta_z-\alpha}{\beta_z-\alpha_z}\cdot p_{I_z}(\vtheta)$ by \eq{pE}.\looseness=-1

Summarizing the foregoing procedures yields Algorithm~\ref{algo:price}, which simultaneously constructs the set $\Z$ and calculates the prices $p_{I_z}(\vtheta)$.
Since it suffices to go down a single path and only perform constant-time computation in each node, the resulting algorithm runs in time $\order(\log \nvals)$, where $\nvals$ denotes the number of distinct values appeared as endpoints of intervals in all the executed transactions.
We defer complete proofs from this paper to the appendix, which is available in the full version of this paper on arXiv. 
\begin{theorem}
	\label{thm:lmsr_price}
	Algorithm~\ref{algo:price} 
    implements $\price$ in time $\order(\log \nvals)$.
\end{theorem}

\begin{algorithm}[t!]
	\small
	\caption{Query price of bundle security for an interval $I=[\alpha,1)$.}
	\label{algo:price}
	\begin{algorithmic}[1]
		\Statex \textbf{Input:}~%
		Interval $I=[\alpha,1)$, $\alpha\in\Omega$, LMSR tree $T$.
		\Statex \textbf{Output:}~%
		Price of bundle security for $I$.
		\medskip
		\State Initialize $z\gets\rootnode$, $P\gets 1$, $\pricevar\gets 0$
		\While {$\alpha_z\ne\alpha$ \textbf{and} $z$ is not a leaf}
		\State $P \gets P e^{s_z/b}$
		\If {$\alpha < \alpha_{\rt{z}}$}
		\State $\pricevar\gets\pricevar + P S_{\rt{z}}/S_{\rootnode}$
		\State $z\gets \lt{z}$
		\Else
		\State $z\gets \rt{z}$
		\EndIf
		\EndWhile
		\State \Return $\pricevar+ \frac{\beta_z-\alpha}{\beta_z-\alpha_z}\cdot
		P S_z/S_{\rootnode}$
	\end{algorithmic}
\end{algorithm}
\vspace{-1ex}

\subsection{Buy Transactions}
We next implement $\buy([\alpha,1),s,T)$ while maintaining the LMSR tree properties.
The main challenge here is to simultaneously maintain \emph{partial-normalization correctness} and \emph{height balance}.
We address this by adapting AVL-tree rebalancing.

We begin by considering the case $\alpha\in\vals(T)$.
Similar to price queries, we conduct binary search for $\alpha$ to obtain the set of nodes $\Z$ that covers $I=[\alpha,1)$.
We update the values of $s_z$ across $z\in\Z$ by adding $s$, and obtain $T'$ that has the same structure as $T$ with the updated share quantities
\[
s'_{z}=
\begin{cases}
s_z+s&\text{if $z\in\Z$}
\\s_z  &\text{otherwise.}
\end{cases}
\]
Thus, the resulting market state is
\[
\vtheta(T')=\sum_{z\in T'} s'_z\vone_{I_z}
=\sum_{z\in T}s_z\vone_{I_z}+\sum_{z\in\Z}{s\vone_{I_z}}
=\vtheta(T) + s\vone_I.
\]
We then rely on the recursive relationship defined in \eq{Sz:rec} to update the partial normalization constants $S_z$.
It suffices to update the ancestors of the nodes $z\in\Z$, all of which lie along the search path to $\alpha$, and each update requires constant time.

When $\alpha\not\in\vals(T)$, we split the leaf $z$ that contains $\alpha\in[\alpha_z,\beta_z)$ before adding shares to $\rt{z}$. 
This may violate the \emph{height-balance property}.
Similar to the AVL insertion algorithm \cite [Section 6.2.3]{Knuth}, we fix any imbalance by means of \emph{rotations}, as we go back along the search path.
Rotations are operations that modify small portions of the tree, and at most two rotations are needed to rebalance the tree~\cite{AVL62}.
We show in Appendix A.2, Lemma 1, that in each rotation, only a constant number of nodes needs to be updated to preserve the \emph{partial-normalization correctness}.
Thus, the overall running time of the $\buy$ operation, presented in Algorithm~\ref{algo:buy}, is $\order(\log \nvals)$.\looseness=-1
%
%

\begin{theorem}
	\label{thm:lmsr_buy}
	Algorithm~\ref{algo:buy} 
    implements $\buy$ in time $\order(\log \nvals)$.
\end{theorem}

\begin{algorithm}[t!]
	\small
	\caption{Buy $s$ shares of bundle security for an interval $I=[\alpha,1)$.}
	\label{algo:buy}
	\begin{algorithmic}[1]
		\algnotext{EndProcedure}
		\Statex \textbf{Input:}~%
		Quantity $s\in\R$, interval $I=[\alpha,1)$, $\alpha\in\Omega$, LMSR tree $T$.
		\Statex \textbf{Output:}~%
		Tree $T$ updated to reflect the purchase of $s$ shares of $\vone_I$.
		\medskip
		\State Define subroutines:
		\Statex \algind \Call{NewLeaf}{$\alpha_0,\beta_0$}:
		return a new leaf node $z$ with
		\Statex \algind\algind
		$I_{z}=[\alpha_0,\beta_0)$,
		$h_{z}=0$, $s_{z}=0$, $S_{z}=(\beta_0-\alpha_0)$
		\Statex \algind \Call{ResetInnerNode}{$z$}:
		reset $h_z$ and $S_z$ based on the children of $z$ 
		\Statex \algind\algind
		$h_{z}\gets 1+\max\set{h_{\lt{z}},h_{\rt{z}}}$,
		$S_z\gets e^{s_z/b}(S_{\lt{z}}+S_{\rt{z}})$
		\Statex \algind \Call{AddShares}{$z,s$}:
		increase the number of shares held in $z$ by $s$
		\Statex \algind\algind
		$s_z\gets s_z+s$, $S_z\gets e^{s/b} S_z$
		\medskip
		\State Initialize $z \gets \rootnode$
		\While {$\alpha_z \neq \alpha$ \textbf{and} $z$ is not a leaf}
		\Comment{add $s$ shares to $z\in\Z$}
		\If {$\alpha < \alpha_{\rt{z}}$}
		\State \Call{AddShares}{$\rt{z},s$}
		\State $z \gets \lt{z}$
		\Else
		\State $z \gets \rt{z}$
		\EndIf
		\EndWhile
		\If {$\alpha_z < \alpha$}
		\Comment{split the leaf $z$}
		\State $\lt{z}\gets{}$\Call{NewLeaf}{$\alpha_z,\alpha$},  $\rt{z}\gets{}$\Call{NewLeaf}{$\alpha,\beta_z$}
		\State $z\gets\rt{z}$
		\EndIf
		\State \Call{AddShares}{$z,s$}
		\While {$z$ is not a $\rootnode$}
		\Comment{trace the binary-search path back} 
		\State $z\gets\pt{z}$
		\If {$\abs{h_{\lt{z}}-h_{\rt{z}}}\ge 2$}
		\Comment{restore height balance}
		\State Rotate $z$ and possibly one of its children
		\Statex ~\hphantom{\textbf{while}} 
		(details in Appendix A.2, Algorithm 5)
		\EndIf
		\State\Call {ResetInnerNode}{$z$}
		\Comment{update $h_z$ and $S_z$}
		\EndWhile
	\end{algorithmic}
\end{algorithm}

\paragraph{Remarks.}
We show that $\price$, $\cost$ and $\buy$ can be implemented in time $\order(\log\nvals)$, which is bounded above by the log of the number of $\buy$ transactions $\order(\log \nbuy)$ and the bit precision of the outcome $\order(\log N)=\order(K)$.%
\footnote{Clearly, $\nvals\le 2\nbuy$ with each $\buy$ transaction introducing at most two new endpoint values.
The value of $\nvals$ is also bounded above by $N+1$ since the interval endpoints are always in $\Omega\cup\set{1}$.}
We note that none of the operations require the knowledge of $K$, so the market in fact supports queries with arbitrary precision.
However, the market precision affects the worst-case loss bound for the market maker, which is $\order(\log N)=\order(K)$.
Next section presents a different construction that achieves a \emph{constant} worst-case loss independent of the market precision.

\section{A Multi-resolution Linearly Constrained Market Maker}
\label{sec:const-loss}
We introduce our second design, referred to as the \emph{multi-resolution linearly constrained market maker} (multi-resolution LCMM).
The design is based on the LMSR, but it enables more flexibility by assigning two or more parallel LMSRs with different liquidity parameters to orchestrate submarkets that offer interval securities at different resolutions.
However, running submarkets independently can create arbitrage opportunities, as any interval expressible in a coarser market can also be expressed in a finer one.
To maintain coherent prices, we design a matrix that imposes linear constraints to tie market prices among different submarkets to support the efficient removal of any arbitrage opportunity, following~\citet{DudikLaPe12}.
We first define the multi-resolution LCMM and its properties, and show that $\price$, $\cost$ and $\buy$ can be implemented in time $\order(\log N)$.\looseness=-1

\subsection{A Multi-resolution LCMM for $[0,1)$}
\subsubsection{A Multi-resolution  Market}
A binary search tree remains at the core of our multi-resolution market construction.
Unlike a log-time LMSR that uses a self-balancing tree, it builds upon a \emph{static} one, where each level of the tree represents a submarket of intervals, forming a finer and finer partition of $[0,1)$.
We start with an example of a market that offers interval securities at two resolutions.\looseness=-1
\newcommand{\Hopen}[1]{\left[#1\right)}
\newcommand{\bigHopen}[1]{\bigl[#1\bigr)}
\newcommand{\BigHopen}[1]{\Bigl[#1\Bigr)}
\begin{example}[Two-level market for $[0, 1)$]
	\label{ex:twolevel}
	We consider a market composed of two submarkets, indexed by $\I_1=\set{11,12}$ and $\I_2=\set{21,22,23,24}$, which partition $[0, 1)$ into interval events at two levels of coarseness:
\vspace{-6pt}
	\begin{gather*}
	I_{11}=\bigHopen{0,\tfrac12},
	I_{12}=\bigHopen{\tfrac12,1};
\\
	I_{21}=\bigHopen{0,\tfrac14},
	I_{22}=\bigHopen{\tfrac14,\tfrac12},
	I_{23}=\bigHopen{\tfrac12,\tfrac34},
	I_{24}=\bigHopen{\tfrac34,1}.
	\end{gather*}
	The market provides six interval securities $\phi_{11},\dotsc,\phi_{24}$ associated with the corresponding interval events, i.e., $\I=\I_1\biguplus\I_2$ and $\card{\I}=6$.
\end{example}
We extend \Ex{twolevel} to multiple resolutions.
We represent the initial independent submarkets with a \emph{complete binary tree} $T^*$ of depth $K$, which corresponds to the bit precision of the outcome $\omega$.
Let $\Z^*$ denote the set of nodes of $T^*$ and $\Z_k$ for $k\in\set{0,1,\dots,K}$ the set of nodes at each level.
$\Z_0$ contains the root associated with $I_{\rootnode}=[0,1)$, and each consecutive level contains the children of nodes from the previous level, which split their corresponding parent intervals in half.
Thus, level $k$ partitions $[0,1)$ into $2^k$ intervals of size $2^{-k}$ and the final level $\Z_K$ contains $N=2^K$ leaves.

We index interval securities by nodes, with their payoffs defined by $\phi_z(\omega)=1\set{\omega\in I_z}$.
We partition securities into submarkets corresponding to levels, i.e., $\I_k = \Z_k$ for $k \leq K$, where $\card{\I_k}=2^k$ and $\I=\biguplus_{k\le K} \I_k$.
For each submarket, we define the LMSR cost function $C_k$ with a \emph{separate} liquidity parameter $b_k>0$: 
\begin{equation}
	\label{eq:level_cost_function}
	C_k(\vtheta_k)=b_k\log\Parens{\sum_{z\in\Z_k} e^{\theta_z/b_k}}.
\end{equation}

%

\subsubsection{A Linearly Constrained  Market Maker}
Following the above multi-resolution construction, the overall market has a \emph{direct-sum cost} $\Cx(\vtheta)=\sum_{k\le K} C_k(\vtheta_k)$,
which corresponds to pricing securities in each block $\I_k$ independently using $C_k$.
However, as there are logical dependencies between securities in different levels, independent pricing may lead to incoherent prices among submarkets and create arbitrage opportunities.
\begin{example}[Arbitrage in a two-level market]
	\label{ex:arbitrage}
	Continuing Example~\ref{ex:twolevel}, we define separate LMSR costs, where $b_1=1$ and $b_2=1$:
	\[
	C_1(\vtheta_1)=\log\Parens{e^{\theta_{11}}\!+e^{\theta_{12}}};\;\;
	C_2(\vtheta_2)=\log\Parens{e^{\theta_{21}}\!+e^{\theta_{22}}\!+e^{\theta_{23}}\!+e^{\theta_{24}}}.
	\]
	The direct-sum market $\Cx(\vtheta)=C_1(\vtheta_1)+C_2(\vtheta_2)$ allows incoherent prices.
	For example, after buying some shares of security $\phi_{21}$ associated with $I_{21}=\bigHopen{0,\tfrac14}$ in submarket $\I_2$, the market can have
	\[
	\px_{11}(\vtheta)=0.5;
	\quad
	\px_{21}(\vtheta)+\px_{22}(\vtheta)=0.6.
	\]
	These prices are incoherent, i.e., do not correspond to probabilities of $I_{11}$, $I_{21}$, $I_{22}$, because under any probability distribution over $\Omega$, we must have $\Pr[I_{11}]=\Pr[I_{21}]+\Pr[I_{22}]$ and $\Pr[I_{12}]=\Pr[I_{23}]+\Pr[I_{24}]$. Thus, a coherent price vector $\vmu\in\R^{\card{\I}}$ must satisfy linear constraints $\mu_{11}-\mu_{21}-\mu_{22}=0$ and $\mu_{12}-\mu_{23}-\mu_{24}=0$, which can be also written as $\va_1^\top\vmu=0$ and $\va_2^\top\vmu=0$ where
	\[
	\va_1=(1,0,-1,-1,0,0)^\top
	\quad \text{and} \quad
	\va_2=(0,1,0,0,-1,-1)^\top.
	\]
	We refer to $\A=(\va_1,\va_2) \in \R^{\card{\I}\times 2}$ as the constraint matrix.
\end{example}
We extend Example~\ref{ex:arbitrage} to specify price constraints in a multi-resolution market. Later we will show how the constraint matrix can be used to remove arbitrage arising from the constraint violations.

Recall that $\M$ denotes a \emph{coherent price space}, where any expected payoff lies in the convex hull of $\{\vphi(\omega)\}_{\omega \in \Omega}$.
For the multi-resolution market, we specify a set of \emph{homogeneous linear equalities} describing a superset of~$\M$.
%
\begin{equation}
\label{eq:Amu=0}
\M\subseteq\set{\vmu\in\R^{\card{\I}}:\:\A^\top\vmu=\vzero}.
\end{equation}

We design the constraint matrix $\A$ to ensure that any pair of submarkets is price coherent, meaning that any interval event $I \subseteq\Omega$ gets the same price on all levels that can express it.
Therefore, for each \emph{inner node} $y \in \Z_l$ where $l<K$, we have
\[
\mu_y=\sum_{z\in\Z_k:\:z\subset y} \mu_z \qquad
\text{for any $l<k\leq K$.}
\]
For algorithmic reasons (as we will see in Section~\ref{sec:lcmm_buy_cost}), we further tie the price of $y$ to the prices of \emph{all} of $y$'s descendants and weight each level by its liquidity parameter $b_k$:
\begin{equation}
\label{eq:constr}
\,\,\underbrace{\!\!\BigParens{\sum_{k>\ell} b_{k}}\!\!}_{B_\ell}\,\, \mu_y=
\sum_{k>\ell}
\BigParens{
	b_k\sum_{z\in\Z_k:\:z\subset y} \mu_z
}.
\vspace{-1ex}
\end{equation}
%

Now we can formally define the constraint matrix $\A$.
Let $\Y^*=\Z^*\wo\Z_K$ be the set of inner nodes of $T^*$ and let $\level(z)$ denote the level of a node~$z$.
The matrix $\A\in\R^{\card{\Z^*}\times\card{\Y^*}}$ contains the constraints from \eq{constr} across all $y\in\Y^*$:
\begin{equation}
\label{eq:A}
A_{zy}=
\begin{cases}
B_{\level(z)}&\text{if $z=y$,}
\\
-b_{\level(z)}&\text{if $z\subset y$,}
\\
0&\text{otherwise.}
\end{cases}
\end{equation}

Arbitrage opportunities arise if the price of bundle $\va_j$ differs from zero, where $\va_j$ denotes the $j$th column of~$\A$.
Traders profit by buying a positive quantity of $\va_j$ if its price is negative, and selling otherwise.
Thus, the constraint matrix $\A$ gives a recipe for arbitrage removal. We provide the intuition
for this in the two-level market, and then give the definition of the multi-resolution LCMM.
\begin{example}[Arbitrage removal in a two-level market]
	\label{ex:lcmm}
	Continuing Example~\ref{ex:arbitrage},
	the prices $\vpx(\vtheta)$ violate the constraint ${\A^\top\vmu=\vzero}$, because $\va_1^\top\vpx(\vtheta)=\px_{11}(\vtheta)-\px_{21}(\vtheta)-\px_{22}(\vtheta) = 0.5-0.6 \neq 0$.
	The vector $\va_1$ reveals an arbitrage opportunity: buy the security $\phi_{11}$ (at the initial price $0.5$) and simultaneously sell securities $\phi_{21}$ and $\phi_{22}$ (at the initial price $0.6$), i.e., buy bundle $\va_1$. Since under any outcome $\omega$, the payout for the bundle $\va_1$ is $0$, this is initially profitable. However, buying $\va_1$ will increase the price of $\phi_{11}$ and decrease the prices of $\phi_{21}$ and $\phi_{22}$. Once a sufficiently large quantity $s$ of shares of $\va_1$ is bought, this form of arbitrage is removed and we have $\va_1^\top\vpx(\smash{\tvtheta})=0$ in a new state $\smash{\tvtheta}=\vtheta+s\va_1=\vtheta+\A\veta$, where $\veta:=(s,0)^\top$.
%
\end{example}
%
A linearly constrained market maker (LCMM)~\cite{DudikLaPe12} leverages violated constraints
similarly as in Example~\ref{ex:lcmm} to remove arbitrage, and then returns the arbitrage
proceeds to the trader. Formally, an LCMM is described by the cost function
\begin{equation}
\label{eq:lcmm}
C(\vtheta) = \inf_{\veta\in\R^{\card{\Y^*}}}\Cx(\vtheta+\A\veta).
\end{equation}
It relies on the direct-sum cost $\Cx$, but with each trader purchase $\vdelta$ that causes incoherent prices, an LCMM automatically seeks the most advantageous cost for the trader by buying bundles $\A\vdelta_\textup{arb}$ on the trader's behalf to remove arbitrage.
\emph{Trader purchases are accumulated as the state $\vtheta$, and automatic purchases made by the LCMM are accumulated as $\A\veta$.}

We note that the purchase of bundle $\A\vdelta_\textup{arb}$ has no effect on the trader's payoff, since $(\A\vdelta_\textup{arb})^\top\vphi(\omega)=0$ for all $\omega\in\Omega$ thanks to \eq{Amu=0} and the fact that $\vphi(\omega)\in\M$.
However, the purchase of $\A\vdelta_\textup{arb}$ can lower the cost, so optimizing over $\vdelta_\textup{arb}$ benefits the traders, while maintaining the same worst-case loss guarantee for the market maker as $\Cx$ \cite{DudikLaPe12}.
Consider a fixed $\vtheta$ and the corresponding $\sveta$ minimizing \eq{lcmm}.
We calculate prices as
$
\vp(\vtheta)=\nabla C(\vtheta)=\nabla\Cx(\vtheta+\A\sveta).
$
By the first order optimality, $\sveta$ minimizes \eq{lcmm} if and only if $\A^\top\bigParens{\nabla\Cx(\vtheta+\A\sveta)}=\vzero$.
This means that $\A^\top\vp(\vtheta)=\vzero$, and thus arbitrage opportunities expressed by $\A$ are completely removed by the LCMM cost function $C$.

To implement an LCMM, we maintain the state $\tvtheta=\vtheta+\A\veta$ in the direct-sum market $\Cx$.
After updating $\vtheta$ to a new value $\vtheta'=\vtheta+\vdelta$, we seek to find $\veta'=\veta+\vdelta_\textup{arb}$ that removes all the arbitrage opportunities expressed by $\A$.
The resulting cost for the trader is
\[
\Cx(\vtheta'+\A\veta')-\Cx(\vtheta+\A\veta)
=\Cx(\tvtheta+\vdelta+\A\vdelta_\textup{arb})
-\Cx(\tvtheta)
.
\]

We finish this section by pointing out two favorable properties of the multi-resolution LCMM.
Above, we have established that LCMM removes all arbitrage opportunities expressed
by $\A$.
The next theorem shows that this actually removes all arbitrage.
The proof shows that consecutive levels are coherent, which by transitivity implies that the overall price vector is coherent (see Appendix A.3).\looseness=-1

\begin{theorem}
	\label{thm:lcmm_arb_free}
	A multi-resolution LCMM is arbitrage-free.
\end{theorem}

The multi-resolution LCMM also enjoys the \emph{bounded-loss} property.
For a suitable choice of liquidities, such as $b_k=\order(1/k^{2.01})$, it can achieve a \emph{constant} worst-case loss bound.
The proof uses the fact that the overall loss is bounded by the sum of losses of level markets, which are at most $b_k\log\,\card{\Z_k}=k b_k\log 2$.
\begin{theorem}
	Let $\set{b_k}_{k=1}^\infty$ be a sequence of positive numbers such that $\sum_{k=1}^\infty kb_k = B^*$ for some finite $B^*$. Then the multi-resolution LCMM with liquidity parameters $b_k$ for $k\le K$ guarantees the worst-case loss of the market maker of at most $B^*\log 2$, regardless of the outcome precision $K$.
	\label{thm:constant_loss}
	\vspace{1ex}
\end{theorem}

\subsubsection{A Multi-resolution LCMM Tree}
We can now formally define the multi-resolution LCMM tree.
The market state of a multi-resolution LCMM is represented by vectors $\vtheta\in\R^{\card{\Z^*}}$
and $\veta\in\R^{\card{\Y^*}}$, whose dimensions can be intractably large (e.g., on the order of $2^K=N$).
However, since each LCMM operation involves only a small set of coordinates of $\vtheta$ and $\veta$, we only keep track of the coordinates accessed so far and represent them as an annotated subtree $T$ of $T^*$, referred to as an \emph{LCMM tree}.
\begin{definition}[LCMM Tree]
	An \emph{LCMM tree} $T$ is a full binary tree, where each node $z$ is annotated with $I_z=[\alpha_z,\beta_z)$, $\theta_z\in\R$, $\eta_z\in\R$, such that
	$I_{\rootnode}=[0,1)$, and for every inner node $z$:
	\[
	\alpha_z=\alpha_{\lt{z}},\quad
	\beta_{\lt{z}}=\alpha_{\rt{z}}=\frac{\alpha_z+\beta_z}{2},\quad
	\beta_{\rt{z}}=\beta_z.
	\]
	\label{def:lcmm_tree}
	\vspace{-2ex}
\end{definition}

The tree $T$ contains the coordinates of $\vtheta$ and $\veta$ accessed so far. Since $\vtheta$ and $\veta$ are initialized to zero, their remaining entries are zero.
We write $\vtheta(T)\in\R^{\card{\Z^*}}$ and $\veta(T)\in\R^{\card{\Y^*}}$ for the vectors represented by $T$.
To calculate prices, we maintain $\veta(T)$ that minimizes \Eq{lcmm}, or equivalently $\veta(T)$ that satisfies
$
\A^\top\vpx\bigParens{\vtheta(T)+\A\veta(T)}=\vzero.
$
If this property holds, we say that an LCMM tree $T$ is \emph{coherent}.


\subsection{Price Queries}
There are many ways to decompose an interval $I$ in a multi-resolution market, but they all yield the same price thanks to coherence.
The no-arbitrage property also guarantees that the price of $[\alpha,\beta)$ can be obtained by subtracting the price of $[\beta,1)$ from $[\alpha,1)$.
Therefore, we focus on pricing one-sided intervals of the form $I=[\alpha,1)$.

Let $T$ be a coherent LCMM tree and $\vtheta\coloneqq\vtheta(T)$ and $\veta\coloneqq\veta(T)$ be the vectors represented by $T$.
Let $\tvtheta=\vtheta+\A\veta$ be the corresponding state in $\Cx$, so the current security prices are $\vmu:=\vpx(\tvtheta)$.
As before, we identify a set of nodes $\Z$ that covers $I$, and then rely on price coherence to calculate each $\mu_z$ along the search path. 

Assume that $z$ is not a root node and we know the price of its parent.
Let $\sib{z}$ denote the sibling of $z$ and $k=\level(z)$.
We can then relate the price of $z$ to the price of $\pt{z}$:
\begin{align}
\label{eq:pxz:1}
\mu_z
&=\frac{\mu_z}{\mu_{\pt{z}}}\cdot\mu_{\pt{z}}
=\frac{\mu_z}{\mu_z+\mu_{\sib{z}}}\cdot\mu_{\pt{z}}
\\
\label{eq:pxz:2}
&
=\frac{e^{\ttheta_z/b_k}}{e^{\ttheta_z/b_k}+e^{\ttheta_{\sib{z}}/b_k}}\cdot\mu_{\pt{z}}.
\end{align}
\eq{pxz:1} follows by price coherence and \eq{pxz:2} follows by the price calculation in \eq{LMSR}.
Thus, we descend the search path to calculate each price $\mu_z$, beginning with $\mu_\rootnode=1$.
It remains to obtain $\ttheta_z$, for which we follow the construction of $\A$ in \eq{A}:
\begin{equation}
\label{eq:ttheta}
\ttheta_z = \theta_z+\sum_{y\in\Y^*} A_{zy}\eta_y
= \theta_z+B_{k}\eta_z-b_{k}\sum_{y\supset z} \eta_y.
\end{equation}
Plugging the above equation back in \eq{pxz:2}, we obtain\footnote{The factor $\exp\braces{-\sum_{y\supset z} \eta_y}=\exp\braces{-\sum_{y\supset\sib{z}} \eta_y}$ appears in both the numerator and the denominator after plugging \eq{ttheta} to \eq{pxz:2}, so it cancels out.}
\begin{equation}
\label{eq:pxz:3}
\mu_z
=\frac{\exp\BigParens{\frac{\theta_z+B_{k}\eta_z}{b_k}}}
{\exp\BigParens{\frac{\theta_z+B_{k}\eta_z}{b_k}}
	+
	\exp\BigParens{\frac{\theta_{\sib{z}}+B_{k}\eta_{\sib{z}}}{b_k}}
} \cdot \mu_{\pt{z}}.
\end{equation}
%

These steps yield Algorithm~\ref{algo:price:lcmm}.
The final line of the algorithm addresses the case when the search ends in the leaf $z$ with $\alpha_z<\alpha<\beta_z$.
Rather than expanding the tree to its lowest level $K$, we use price coherence again: since any strict descendant $z'\subset z$ on the path from $z$ to a leaf node $u \in \Z_K$ has $\theta_{z'}=\eta_{z'}=0$ by market initialization, all leaf nodes have the same price. 
Therefore, the price of $[\alpha,\beta_z)$ equals $\frac{\beta_z-\alpha}{\beta_z-\alpha_z}\cdot\mu_z$.

The length of search path for $\alpha$ is $\prec(\alpha)$, which denotes the bit precision of $\alpha$, defined as the smallest integer $k$ such that $\alpha$ is an integer multiple of $2^{-k}$.
As the computation at each node only requires constant time, the time to price $I = [\alpha, 1)$ is $\order(\prec(\alpha))$, which is bounded above by $\order(K)$.

\begin{theorem}
	Let $I=[\alpha,1)$, $\alpha\in\Omega$.
	Algorithm~\ref{algo:price:lcmm} implements $\price(I,T)$ in time $\order(\prec(\alpha))$.
	\label{thm:lcmm_price}
\end{theorem}
\begin{algorithm}[t!]
	\small
	\caption{Query price of bundle security for an interval $I=[\alpha,1)$.}
	\label{algo:price:lcmm}
	\begin{algorithmic}[1]
		\Statex \textbf{Input:}~%
		Interval $I=[\alpha,1)$, $\alpha\in\Omega$,
		coherent LCMM tree $T$.
		\Statex \textbf{Output:}~%
		Price of bundle security for $I$.
		\medskip
		\State Initialize $z \gets \rootnode$, $\mu_z\gets 1$, $\pricevar\gets 0$
		\While {$\alpha_z \neq \alpha$ \textbf{and} $z$ is not a leaf}
		\State $z_l\gets\lt{z}$, $z_r\gets\rt{z}$, $k\gets\level(z_l)$
		\State $e_l\gets\exp\set{(\theta_{z_l}+B_{k}\eta_{z_l})/b_k}$,
		$e_r\gets\exp\set{(\theta_{z_r}+B_{k}\eta_{z_r})/b_k}$,
		\Statex ~\hphantom{\textbf{wi}} $\mu_{z_l}\gets\frac{e_l}{e_l+e_r}\mu_z$,
		$\mu_{z_r}\gets\frac{e_r}{e_l+e_r}\mu_z$
		\Comment{calculate prices by \Eq{pxz:3}}
		\If {$\alpha < \alpha_{\rt{z}}$}
		\State $z \gets z_l, \quad$$\pricevar \gets \pricevar + \mu_{z_r}$
		\Else
		\State $z \gets z_r$
		\EndIf
		\EndWhile
		\State \Return $\pricevar+ \frac{\beta_z-\alpha}{\beta_z-\alpha_z}\cdot\mu_z$
	\end{algorithmic}
\end{algorithm}

\subsection{Buy and Cost Operations}
\label{sec:lcmm_buy_cost}
\begin{algorithm}[t!]
	\small
	\caption{Buy $s$ shares of bundle security for an interval $I=[\alpha,1)$.}
	\label{algo:buy:lcmm}
	\begin{algorithmic}[1]
		\algnotext{EndProcedure}
		\algnotext{EndFunction}
		\Statex \textbf{Input:}~%
		Quantity $s\mkern-5mu\in\mkern-3mu\R$, interval $I\!=\![\alpha\mkern-1mu,\mkern-2mu 1)$, $\alpha\mkern-5mu\in\mkern-3mu\Omega$,
		coherent LCMM tree $T$.
		\Statex \textbf{Output:}~%
		Cost of $s$ shares of bundle security for $I$, the updated tree $T$.
		\medskip
		\State Define subroutines:
		\Statex \algind \Call{NewLeaf}{$\alpha_0,\beta_0$}:
		return a new leaf node $z$ with
		\Statex \algind\algind
		$I_{z}=[\alpha_0,\beta_0)$,
		$\theta_{z}=0$, $\eta_{z}=0$
		\medskip
		\Statex \algind \Call{RemoveArbitrage}{$y,\muother$}: restore price coherence among
		\Statex ~\hphantom{De} submarkets $k \geq \level(y)$ following \Eq{arb_t} and update cost
		\Statex \algind\algind Let $\ell=\level(y)$, $y'=\sib{y}$, $t= \frac{b_\ell}{B_{\ell-1}}\log\Parens{\frac{1-\mu_y}{\mu_y}\cdot\frac{\muother}{1-\muother}}$
		\Statex \algind\algind $S=\mu_y e^{tB_\ell/b_\ell}+1-\mu_y$, $\Sother=\muother e^{-t}+1-\muother$
		\Statex \algind\algind $\eta_y \gets \eta_y + t$, $\mu_y\gets \mu_y e^{tB_\ell/b_\ell}/S$, $\mu_{y'}\gets \mu_{y'}/S$
		\Statex \algind\algind $\costvar\gets\costvar + (b_\ell\log S) + (B_\ell\log\Sother)$
		\medskip
		\Statex \algind \Call{AddShares}{$z,s$}:
		increase shares held in $z$ by $s$, update cost, and
		\Statex ~\hphantom{De} restore price coherence among submarkets $k \geq \level(z)$
		\Statex \algind\algind Let $\ell=\level(z)$, $z'=\sib{z}$, $\muother=\mu_z$, $S = \mu_z e^{s/b_\ell}+1-\mu_z$
		\Statex \algind\algind $\theta_{z} \gets \theta_{z}+s$
		\Statex \algind\algind $\costvar\gets\costvar + (b_\ell\log S)$
		\Statex \algind\algind $\mu_z\gets \mu_z e^{s/b_\ell}/S$, $\mu_{z'}\gets \mu_{z'}/S$
		\Statex \algind\algind \Call{RemoveArbitrage}{$z,\muother$}
		\medskip
		\State Initialize $z \gets \rootnode$, $\mu_z\gets 1$, a global variable $\costvar\gets 0$
		\While {$\alpha_z \neq \alpha$}
		\If {$z$ is a leaf}
		\State $\lt{z}\gets{}$\Call{NewLeaf}{$\alpha_z,\,\frac12(\alpha_z+\beta_z)$},
		\Statex ~\hphantom{\textbf{whiiw}}$\rt{z}\gets{}$\Call{NewLeaf}{$\frac12(\alpha_z+\beta_z),\,\beta_z$}
		\EndIf
		\State Calculate $\mu_{\lt{z}}$, $\mu_{\rt{z}}$, and update $z$ according to $\alpha$
        \Statex\algind (same as \Algo{price:lcmm} lines 3-8)
		\EndWhile
		\State \Call{AddShares}{$z,s$}
		\While {$z$ is not a $\rootnode$}
		\Comment{remove arbitrage up the search path}
		\State $z'\gets\sib{z}$, $y\gets\pt{z}$
		\If {$z'=\rt{y}$}
		\State \Call{AddShares}{$z',s$}
		\Comment{add shares to $z \in \Z$}
		\EndIf
		\State \Call{RemoveArbitrage}{$y$, $\mu_z+\mu_{z'}$}
		\State $z\gets y$
		\EndWhile
		\State \Return $\costvar$
%
	\end{algorithmic}
\end{algorithm}

Different from LMSR, the cost query for a multi-resolution LCMM cannot be directly derived from prices.
We instead augment $\buy$ to implement $\cost$ by executing $\buy$ and then reverting all the changes.
We focus on $\buy(I,s,T)$ for $I=[\alpha,1)$.
By buying $s$ shares of $[\alpha,1)$ and then $(-s)$ shares of $[\beta,1)$, we obtain buying $[\alpha,\beta)$.

We summarize the procedure in \Algo{buy:lcmm}, which performs $\buy(I,s,T)$ and keeps track of $\cost(I,s,T)$.
Similar to price queries, we start with a set of nodes $\Z$ that partition $I$, by searching for $\alpha$ and simultaneously calculating prices $\mu_z$ along the way (lines 3--6).

We then proceed back up the search path, adding $s$ shares to nodes within the cover $\Z$ (lines 7--13).
Consider one of such node $y\in\Z$ at level $\ell\coloneqq\level(y)$.
Increasing $\theta_y$ by $s$ creates price incoherence between the submarket at level $\ell$ and submarkets at all other levels.
We design \textsc{RemoveArbitrage} to remove any arbitrage opportunity between level $\ell$ and \emph{all finer levels} with $k>\ell$.
We show in Appendix A.6, Lemma 3, that in order to restore coherence, it suffices to update $\eta_y$ by a closed-form amount:
\begin{equation}
\label{eq:arb_t}
t= \frac{b_\ell}{B_{\ell-1}}\log\Parens{\frac{1-\mu_y}{\mu_y}\cdot\frac{\muother}{1-\muother}},
\end{equation}
where $\muother = \mu_{\lt{y}}+\mu_{\rt{y}}$ records the price of $y$ in all the finer levels.
This key algorithmic step is enabled by the arbitrage bundle $\va_y$, which corresponds to buying $\phi_y$ on the level $\ell$ while selling securities associated with all descendants of $y$, with their shares appropriately weighted by the respective liquidity values as specified in the constraint matrix $\A$.

The market remains incoherent between $\ell$ and \emph{all coarser levels} $k<\ell$.
Since the updates have been localized to the subtree rooted at $y$, we use Lemma 3 again to update $\eta_{\pt{y}}$ and restore coherence among all levels $k\ge\ell-1$ (line 12).
We continue in this manner back along the path to root to restore a coherent market.

The algorithm also tracks the total cost of the $\buy$ transaction by evaluating \eq{costE} in the component submarkets.
Note that costs in all submarkets with $k>\ell$ can be evaluated simultaneously thanks to the restored coherence.
Since the computations in each accessed node are constant time,
Algorithm~\ref{algo:buy:lcmm} runs in time $\order(\prec(\alpha))$.\looseness=-1

\begin{theorem}
	Let $I=[\alpha,1)$, $\alpha\in\Omega$.
	Algorithm~\ref{algo:buy:lcmm} implements a simultaneous $\buy(I,s,T)$ and $\cost(I,s,T)$ in time $\order(\prec(\alpha))$.
	\label{thm:lcmm_buy}
	\vspace{1ex}
\end{theorem}

\paragraph{Remarks.}
In Algorithms \ref{algo:price:lcmm} and \ref{algo:buy:lcmm}, we assume that each node $z$ can store a scalar $\mu_z$, which can be modified during the run to support price calculations but is disposed afterwards.
The only part of our algorithms that depends on $K$ are the cumulative liquidities $B_\ell=\sum_{k=\ell+1}^K b_k$.
To remove such dependence, we can use $B'_\ell=\smash{\sum_{k=\ell+1}^\infty b_k=B^*-\sum_{k=1}^\ell b_k}$, where $B^*=\sum_{k=1}^\infty b_k$. This has no impact on the correctness of our algorithms: if at a given time the largest level in the tree $T$ is $L$, we can simply view $T$ as a multi-resolution LCMM with $K=L+1$ and liquidities
$b_1,b_2,\dotsc,b_L,B'_L$.
The last level $K=L+1$ then corresponds to infinitely many mutually coherent markets $\set{C_k}_{k=L+1}^\infty$.
Thus, a multi-resolution LCMM can achieve a constant loss bound regardless of $K$ and support market operations for $I=[\alpha,\beta)$ in time $\order(\prec(\alpha)+\prec(\beta))$. 

\section{Discussion and Illustration}
\label{sec:exp}
We have proposed two cost-function-based market makers that support trading interval securities of arbitrary precision and execute market operations exponentially faster than previous designs.
In what situations is one preferable over the other?

The log-time LMSR enjoys better storage and runtime efficiency, because search paths in LMSR tree are shorter thanks to its height-balance property.
The log-time LMSR would therefore be computationally preferable, for example, when the designer expects betting interest to be concentrated on a smaller set of intervals.
However, the log-time LMSR implements a standard LMSR, which faces well-known design challenges, such as the requirement to set a suitable liquidity value and the precision of bets in advance.
Correctly setting these parameters often requires a good estimate of trader interest even before trading in the market starts.


On the other hand, the multi-resolution LCMM does not require a hard specification of the betting precision. Flexible pricing allows the designer to attenuate liquidity across different precisions in a way that best reflects the designer's information-gathering priorities.
%
For example, an LMSR that operates at precision $k=4$ with liquidity $b$ can be represented by an LCMM with the level liquidity values $\vb = (0, 0, 0, b, 0, 0, \dotsc)$.
Moreover, if the market designer expects most of the information at precision 4 but also wants to support bets up to precision 8, they could run an LCMM with the liquidity placed at two levels as $\vb = (0, 0, 0, b_4, 0, 0, 0, b_8)$.
By choosing different values $b_4$ and $b_8$, the market designer can express utility for information at different precision levels.

We empirically highlight such flexibility by showing how LCMM can interpolate between LMSRs at different resolutions, allowing the market to match the coarseness of traders' information.
We conduct agent-based simulation using the trader model with exponential utility and exponential-family beliefs~\cite{Abernethy2014,Dudik2017}.%
\footnote{We note that while our market makers support agents with any beliefs and utilities, the exponential trader model is convenient, as it allows a closed-form derivation of \emph{market-clearing price}~\cite{Abernethy2014,Dudik2017}, which can be viewed as a ``ground truth'' for the information elicitation.}
We defer the detailed trader model to Appendix B.1.
Agents trade with either an LMSR or a multi-resolution LCMM, and we are interested in evaluating market makers' performance in terms of \textit{price convergence error}, calculated as the relative entropy between the \textit{market-clearing price} (that is the price reached when agents only trade among themselves) and the price maintained by the market maker.
We operate in a market over $[0,1)$ and the outcome is specified with $K=10$ bits.
We consider budget-limited market makers, whose worst-case loss may not exceed a budget constraint $B$.
For LMSR at precision $k$, this means setting the liquidity parameter to $b=B/\log\parens{2^k}$.
Following our motivating example, we compare two LMSR markets at precision levels 4 and 8, denoted as $\texttt{LMSR}_{k=4}$ and $\texttt{LMSR}_{k=8}$, to an LCMM that evenly splits budget to precision levels 4 and 8, denoted as $\texttt{LCMM}_{50/50}$.%
\footnote{The LCMM has an infinite number of choices for its liquidity at each level. We choose $\texttt{LCMM}_{50/50}$ as an instance here to showcase its interpolation ability.}

\begin{figure}[t!]
	\centering
	\begin{subfigure}[t]{0.23\textwidth}
		\centering
		\includegraphics[width=\textwidth]{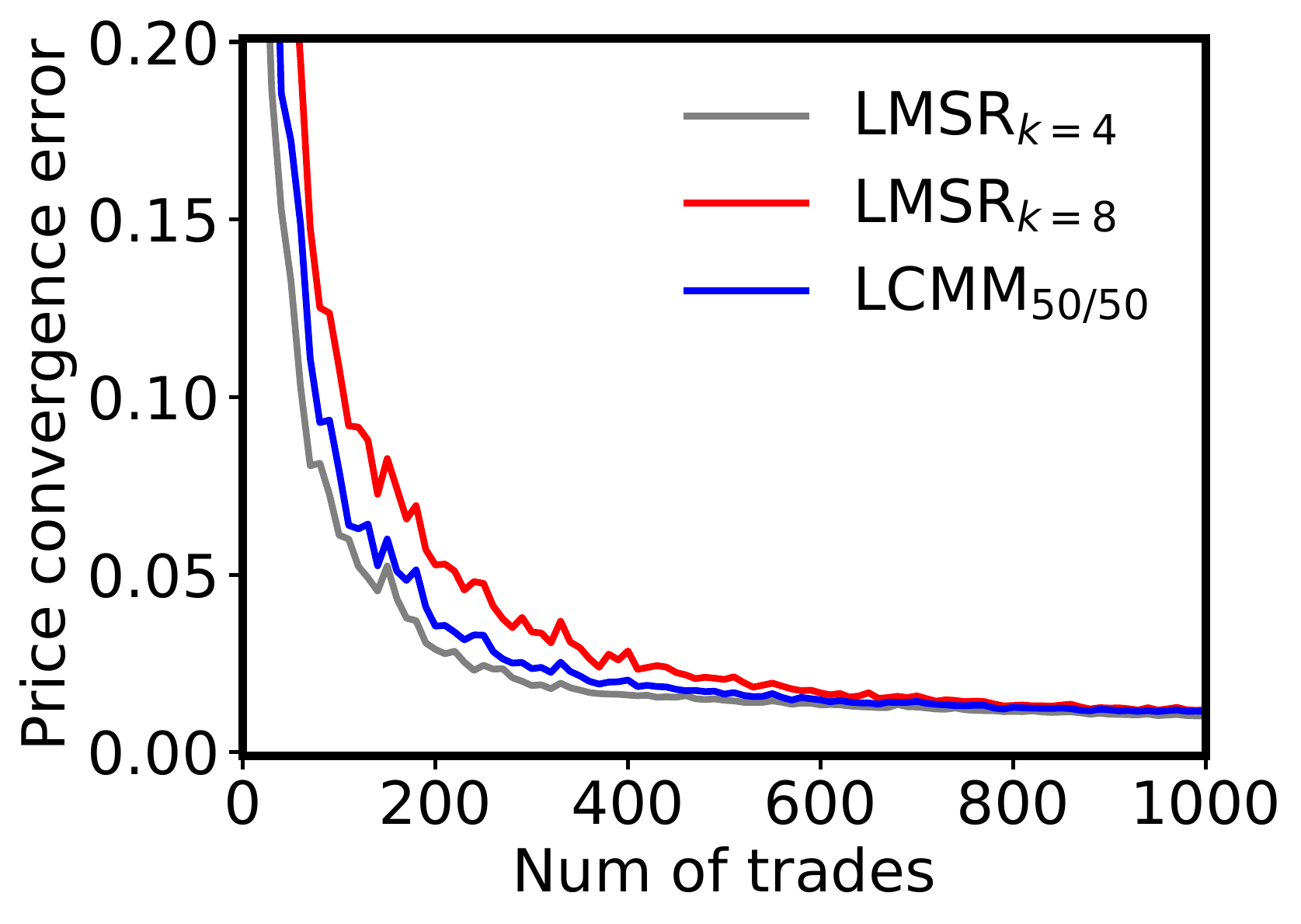}
		\caption{$k=4$.}
	\end{subfigure}
	\begin{subfigure}[t]{0.23\textwidth}
		\centering
		\includegraphics[width=\textwidth]{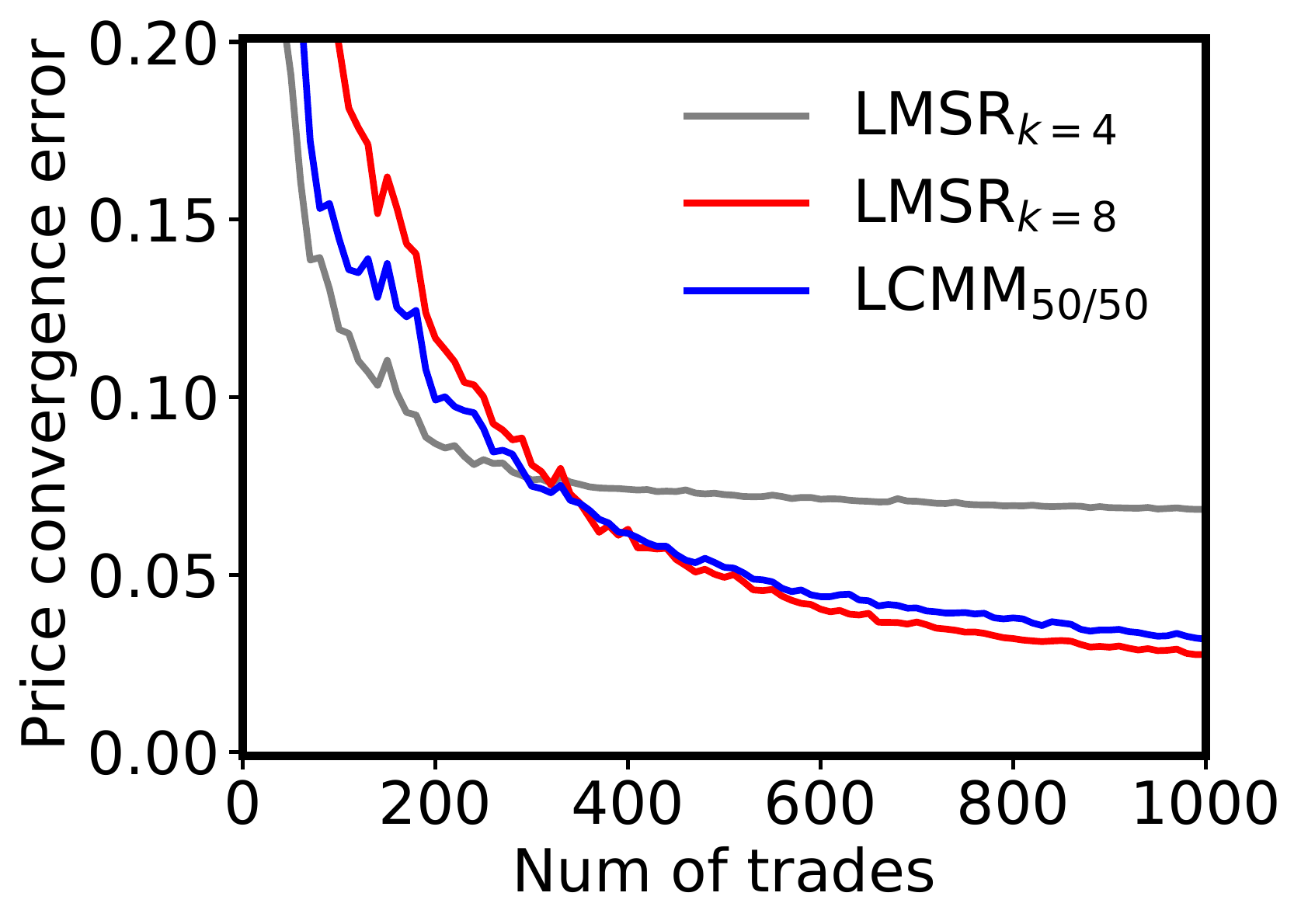}
		\caption{$k=8$.}
	\end{subfigure}
	\caption{The price convergence error as a function of the number of trades, measured at two resolution levels.}
	\label{fig:numerical_exp}
\end{figure}

Fig.~\ref{fig:numerical_exp} shows the price convergence as a function of the number of trades.
As one may expect, $\texttt{LMSR}_{k=4}$ achieves a faster price convergence at the coarser precision level $k=4$ compared to $\texttt{LMSR}_{k=8}$ (Fig.~\ref{fig:numerical_exp}a), but fails to elicit information at any finer granularity by design.\footnote{In Fig.~\ref{fig:numerical_exp}b, to facilitate comparisons, we assume that
	$\texttt{LMSR}_{k=4}$ splits the price of a coarse interval evenly into prices of finer intervals.}
The proposed $\texttt{LCMM}_{50/50}$, by equally splitting the budget between $k=4$ and $k=8$, is able to interpolate between the performance of $\texttt{LMSR}_{k=4}$ and $\texttt{LMSR}_{k=8}$ and achieves the ``best of both worlds'': it can elicit forecasts at the finer level $k=8$ similarly to $\texttt{LMSR}_{k=8}$, but also obtain a fast convergence at the coarser level $k=4$, almost matching the convergence speed of $\texttt{LMSR}_{k=4}$. 

Two immediate questions arise from our work.
First, do our constructions generalize to two- or higher-dimensional outcomes?
One promising avenue is to combine the ideas from our log-time LMSR market maker
with multi-dimensional segment trees~\cite{Mishra16} to obtain an efficient multi-dimensional LMSR based on a static tree.
However, it is not clear how to generalize our balanced LMSR tree construction or the multi-resolution LCMM.
%
%
Second, does our approach extend to non-interval securities, such as call options?
We leave these questions open for future research.


\balance
\bibliographystyle{ACM-Reference-Format}
\bibliography{refs}


\begin{thebibliography}{24}


\ifx \showCODEN    \undefined \def \showCODEN     #1{\unskip}     \fi
\ifx \showDOI      \undefined \def \showDOI       #1{#1}\fi
\ifx \showISBNx    \undefined \def \showISBNx     #1{\unskip}     \fi
\ifx \showISBNxiii \undefined \def \showISBNxiii  #1{\unskip}     \fi
\ifx \showISSN     \undefined \def \showISSN      #1{\unskip}     \fi
\ifx \showLCCN     \undefined \def \showLCCN      #1{\unskip}     \fi
\ifx \shownote     \undefined \def \shownote      #1{#1}          \fi
\ifx \showarticletitle \undefined \def \showarticletitle #1{#1}   \fi
\ifx \showURL      \undefined \def \showURL       {\relax}        \fi
\providecommand\bibfield[2]{#2}
\providecommand\bibinfo[2]{#2}
\providecommand\natexlab[1]{#1}
\providecommand\showeprint[2][]{arXiv:#2}

\bibitem[\protect\citeauthoryear{Abernethy, Chen, and Vaughan}{Abernethy
  et~al\mbox{.}}{2011}]%
        {AbernethyChVa11}
\bibfield{author}{\bibinfo{person}{Jacob Abernethy}, \bibinfo{person}{Yiling
  Chen}, {and} \bibinfo{person}{Jennifer~Wortman Vaughan}.}
  \bibinfo{year}{2011}\natexlab{}.
\newblock \showarticletitle{An optimization-based framework for automated
  market-making}. In \bibinfo{booktitle}{\emph{Proceedings of the 12th ACM
  Conference on Electronic Commerce}}.
\newblock


\bibitem[\protect\citeauthoryear{Abernethy, Kutty, Lahaie, and Sami}{Abernethy
  et~al\mbox{.}}{2014}]%
        {Abernethy2014}
\bibfield{author}{\bibinfo{person}{Jacob Abernethy}, \bibinfo{person}{Sindhu
  Kutty}, \bibinfo{person}{S{\'e}bastien Lahaie}, {and} \bibinfo{person}{Rahul
  Sami}.} \bibinfo{year}{2014}\natexlab{}.
\newblock \showarticletitle{Information aggregation in exponential family
  markets}. In \bibinfo{booktitle}{\emph{Proceedings of the 15th ACM Conference
  on Economics and Computation}}. \bibinfo{pages}{395--412}.
\newblock


\bibitem[\protect\citeauthoryear{Adel{$'\!$}son-Vel{$'\!$}ski\u{\i} and
  Landis}{Adel{$'\!$}son-Vel{$'\!$}ski\u{\i} and Landis}{1962}]%
        {AVL62}
\bibfield{author}{\bibinfo{person}{G.~M. Adel{$'\!$}son-Vel{$'\!$}ski\u{\i}}
  {and} \bibinfo{person}{E.~M. Landis}.} \bibinfo{year}{1962}\natexlab{}.
\newblock \showarticletitle{An algorithm for the organization of information}.
\newblock \bibinfo{journal}{\emph{Soviet Mathematics---Doklady}}
  \bibinfo{volume}{3} (\bibinfo{year}{1962}), \bibinfo{pages}{1259--1263}.
\newblock


\bibitem[\protect\citeauthoryear{Chakraborty, Das, Lavoie, Magdon-Ismail, and
  Naamad}{Chakraborty et~al\mbox{.}}{2013}]%
        {Chakraborty13}
\bibfield{author}{\bibinfo{person}{Mithun Chakraborty}, \bibinfo{person}{Sanmay
  Das}, \bibinfo{person}{Allen Lavoie}, \bibinfo{person}{Malik Magdon-Ismail},
  {and} \bibinfo{person}{Yonatan Naamad}.} \bibinfo{year}{2013}\natexlab{}.
\newblock \showarticletitle{Instructor rating markets}. In
  \bibinfo{booktitle}{\emph{Proceedings of the 27th AAAI Conference on
  Artificial Intelligence}}. \bibinfo{pages}{159--165}.
\newblock


\bibitem[\protect\citeauthoryear{Chen, Fortnow, Lambert, Pennock, and
  Vaughan}{Chen et~al\mbox{.}}{2008a}]%
        {ChenEtAl08}
\bibfield{author}{\bibinfo{person}{Yiling Chen}, \bibinfo{person}{Lance
  Fortnow}, \bibinfo{person}{Nicolas Lambert}, \bibinfo{person}{David~M.
  Pennock}, {and} \bibinfo{person}{Jennifer~Wortman Vaughan}.}
  \bibinfo{year}{2008}\natexlab{a}.
\newblock \showarticletitle{Complexity of combinatorial market makers}. In
  \bibinfo{booktitle}{\emph{Proceedings of the 9th ACM Conference on Electronic
  Commerce}}.
\newblock


\bibitem[\protect\citeauthoryear{Chen, Fortnow, Nikolova, and Pennock}{Chen
  et~al\mbox{.}}{2007}]%
        {Chen:07}
\bibfield{author}{\bibinfo{person}{Yiling Chen}, \bibinfo{person}{Lance
  Fortnow}, \bibinfo{person}{Evdokia Nikolova}, {and} \bibinfo{person}{David~M.
  Pennock}.} \bibinfo{year}{2007}\natexlab{}.
\newblock \showarticletitle{Betting on permutations}. In
  \bibinfo{booktitle}{\emph{Proceedings of the 8th ACM Conference on Electronic
  Commerce}}. \bibinfo{pages}{326--335}.
\newblock


\bibitem[\protect\citeauthoryear{Chen, Goel, and Pennock}{Chen
  et~al\mbox{.}}{2008b}]%
        {Chen:08b}
\bibfield{author}{\bibinfo{person}{Yiling Chen}, \bibinfo{person}{Sharad Goel},
  {and} \bibinfo{person}{David~M. Pennock}.} \bibinfo{year}{2008}\natexlab{b}.
\newblock \showarticletitle{Pricing combinatorial markets for tournaments}. In
  \bibinfo{booktitle}{\emph{Proceedings of the 40th Annual ACM Symposium on
  Theory of Computing}}. \bibinfo{pages}{305--314}.
\newblock


\bibitem[\protect\citeauthoryear{Chen and Pennock}{Chen and Pennock}{2007}]%
        {ChenPe07}
\bibfield{author}{\bibinfo{person}{Yiling Chen} {and} \bibinfo{person}{David~M.
  Pennock}.} \bibinfo{year}{2007}\natexlab{}.
\newblock \showarticletitle{A utility framework for bounded-loss market
  makers}. In \bibinfo{booktitle}{\emph{Proceedings of the 23rd Conference on
  Uncertainty in Artificial Intelligence}}.
\newblock


\bibitem[\protect\citeauthoryear{Cormen, Leiserson, and Rivest}{Cormen
  et~al\mbox{.}}{1999}]%
        {CLR99}
\bibfield{author}{\bibinfo{person}{Thomas~H. Cormen},
  \bibinfo{person}{Charles~E. Leiserson}, {and} \bibinfo{person}{Ronald~L.
  Rivest}.} \bibinfo{year}{1999}\natexlab{}.
\newblock \bibinfo{booktitle}{\emph{Introduction to Algorithms}}.
\newblock \bibinfo{publisher}{The MIT Press}.
\newblock


\bibitem[\protect\citeauthoryear{Dud\'ik, Lahaie, and Pennock}{Dud\'ik
  et~al\mbox{.}}{2012}]%
        {DudikLaPe12}
\bibfield{author}{\bibinfo{person}{Miroslav Dud\'ik},
  \bibinfo{person}{S\'ebastien Lahaie}, {and} \bibinfo{person}{David~M.
  Pennock}.} \bibinfo{year}{2012}\natexlab{}.
\newblock \showarticletitle{A tractable combinatorial market maker using
  constraint generation}. In \bibinfo{booktitle}{\emph{Proceedings of the 13th
  ACM Conference on Electronic Commerce}}.
\newblock


\bibitem[\protect\citeauthoryear{Dud{\'i}k, Lahaie, Rogers, and
  Wortman~Vaughan}{Dud{\'i}k et~al\mbox{.}}{2017}]%
        {Dudik2017}
\bibfield{author}{\bibinfo{person}{Miroslav Dud{\'i}k},
  \bibinfo{person}{S\'ebastien Lahaie}, \bibinfo{person}{Ryan~M Rogers}, {and}
  \bibinfo{person}{Jennifer Wortman~Vaughan}.} \bibinfo{year}{2017}\natexlab{}.
\newblock \showarticletitle{A decomposition of forecast error in prediction
  markets}.
\newblock In \bibinfo{booktitle}{\emph{Advances in Neural Information
  Processing Systems}}. \bibinfo{pages}{4371--4380}.
\newblock


\bibitem[\protect\citeauthoryear{Gao, Chen, and Pennock}{Gao
  et~al\mbox{.}}{2009}]%
        {GaoChenPennock09}
\bibfield{author}{\bibinfo{person}{Xi Gao}, \bibinfo{person}{Yiling Chen},
  {and} \bibinfo{person}{David~M. Pennock}.} \bibinfo{year}{2009}\natexlab{}.
\newblock \showarticletitle{Betting on the real line}. In
  \bibinfo{booktitle}{\emph{Proceedings of the 5th Workshop on Internet and
  Network Economics}}.
\newblock


\bibitem[\protect\citeauthoryear{Guo and Pennock}{Guo and Pennock}{2009}]%
        {Guo:09}
\bibfield{author}{\bibinfo{person}{Mingyu Guo} {and} \bibinfo{person}{David~M.
  Pennock}.} \bibinfo{year}{2009}\natexlab{}.
\newblock \showarticletitle{Combinatorial prediction markets for event
  hierarchies}. In \bibinfo{booktitle}{\emph{Proceedings of the 8th
  International Conference on Autonomous Agents and Multiagent Systems}}.
  \bibinfo{pages}{201--208}.
\newblock


\bibitem[\protect\citeauthoryear{Hanson}{Hanson}{1999}]%
        {hanson1999}
\bibfield{author}{\bibinfo{person}{Robin~D. Hanson}.}
  \bibinfo{year}{1999}\natexlab{}.
\newblock \showarticletitle{Decision markets}.
\newblock \bibinfo{journal}{\emph{IEEE Intelligent Systems}}
  \bibinfo{volume}{14}, \bibinfo{number}{3} (\bibinfo{year}{1999}),
  \bibinfo{pages}{16--19}.
\newblock


\bibitem[\protect\citeauthoryear{Hanson}{Hanson}{2003}]%
        {Hanson03}
\bibfield{author}{\bibinfo{person}{Robin~D. Hanson}.}
  \bibinfo{year}{2003}\natexlab{}.
\newblock \showarticletitle{Combinatorial information market design}.
\newblock \bibinfo{journal}{\emph{Information Systems Frontiers}}
  \bibinfo{volume}{5}, \bibinfo{number}{1} (\bibinfo{year}{2003}),
  \bibinfo{pages}{107--119}.
\newblock


\bibitem[\protect\citeauthoryear{Hanson}{Hanson}{2007}]%
        {Hanson07}
\bibfield{author}{\bibinfo{person}{Robin~D. Hanson}.}
  \bibinfo{year}{2007}\natexlab{}.
\newblock \showarticletitle{Logarithmic market scoring rules for modular
  combinatorial information aggregation}.
\newblock \bibinfo{journal}{\emph{Journal of Prediction Markets}}
  \bibinfo{volume}{1}, \bibinfo{number}{1} (\bibinfo{year}{2007}),
  \bibinfo{pages}{1--15}.
\newblock


\bibitem[\protect\citeauthoryear{Knuth}{Knuth}{1998}]%
        {Knuth}
\bibfield{author}{\bibinfo{person}{Donald~E. Knuth}.}
  \bibinfo{year}{1998}\natexlab{}.
\newblock \bibinfo{booktitle}{\emph{The Art of Computer Programming, Volume 3:
  Sorting and Searching}}.
\newblock \bibinfo{publisher}{Addison Wesley}.
\newblock


\bibitem[\protect\citeauthoryear{Laskey, Sun, Hanson, Twardy, Matsumoto, and
  Goldfedder}{Laskey et~al\mbox{.}}{2018}]%
        {LaskeyEtAl18}
\bibfield{author}{\bibinfo{person}{Kathryn~Blackmond Laskey},
  \bibinfo{person}{Wei Sun}, \bibinfo{person}{Robin~D. Hanson},
  \bibinfo{person}{Charles Twardy}, \bibinfo{person}{Shou Matsumoto}, {and}
  \bibinfo{person}{Brandon Goldfedder}.} \bibinfo{year}{2018}\natexlab{}.
\newblock \showarticletitle{Graphical model market maker for combinatorial
  prediction markets}.
\newblock \bibinfo{journal}{\emph{Journal of Artificial Intelligence Research}}
   \bibinfo{volume}{63} (\bibinfo{year}{2018}), \bibinfo{pages}{421--460}.
\newblock


\bibitem[\protect\citeauthoryear{Mishra}{Mishra}{2016}]%
        {Mishra16}
\bibfield{author}{\bibinfo{person}{Pushkar Mishra}.}
  \bibinfo{year}{2016}\natexlab{}.
\newblock \showarticletitle{On Updating and Querying Sub-arrays of
  Multidimensional Arrays}.
\newblock \bibinfo{journal}{\emph{CoRR}}  \bibinfo{volume}{abs/1311.6093}
  (\bibinfo{year}{2016}).
\newblock


\bibitem[\protect\citeauthoryear{Othman, Pennock, Reeves, and Sandholm}{Othman
  et~al\mbox{.}}{2013}]%
        {Othman13}
\bibfield{author}{\bibinfo{person}{Abraham Othman}, \bibinfo{person}{David~M.
  Pennock}, \bibinfo{person}{Daniel~M. Reeves}, {and} \bibinfo{person}{Tuomas
  Sandholm}.} \bibinfo{year}{2013}\natexlab{}.
\newblock \showarticletitle{A practical liquidity-sensitive automated market
  maker}.
\newblock \bibinfo{journal}{\emph{ACM Transactions on Economics and
  Computation}} \bibinfo{volume}{1}, \bibinfo{number}{3}
  (\bibinfo{year}{2013}), \bibinfo{pages}{14:1--14:25}.
\newblock


\bibitem[\protect\citeauthoryear{Othman and Sandholm}{Othman and
  Sandholm}{2010}]%
        {Othman10}
\bibfield{author}{\bibinfo{person}{Abraham Othman} {and}
  \bibinfo{person}{Tuomas Sandholm}.} \bibinfo{year}{2010}\natexlab{}.
\newblock \showarticletitle{Automated market-making in the large: The Gates
  Hillman Prediction Market}. In \bibinfo{booktitle}{\emph{Proceedings of the
  11th ACM Conference on Electronic Commerce}}. \bibinfo{pages}{367--376}.
\newblock


\bibitem[\protect\citeauthoryear{Othman and Sandholm}{Othman and
  Sandholm}{2012}]%
        {Othman12}
\bibfield{author}{\bibinfo{person}{Abraham Othman} {and}
  \bibinfo{person}{Tuomas Sandholm}.} \bibinfo{year}{2012}\natexlab{}.
\newblock \showarticletitle{Automated market makers that enable new settings:
  Extending constant-utility cost functions}. In
  \bibinfo{booktitle}{\emph{Auctions, Market Mechanisms, and Their
  Applications}}. \bibinfo{pages}{19--30}.
\newblock


\bibitem[\protect\citeauthoryear{Plott and Chen}{Plott and Chen}{2002}]%
        {Chen2002}
\bibfield{author}{\bibinfo{person}{Charles~R. Plott} {and}
  \bibinfo{person}{Kay-Yut Chen}.} \bibinfo{year}{2002}\natexlab{}.
\newblock \bibinfo{title}{Information aggregation mechanisms: Concept, design
  and implementation for a sales forecasting problem}.  (\bibinfo{year}{2002}).
\newblock
\newblock
\shownote{Working paper No. 1131, California Institute of Technology.}


\bibitem[\protect\citeauthoryear{Xia and Pennock}{Xia and Pennock}{2011}]%
        {XiaPe11}
\bibfield{author}{\bibinfo{person}{Lirong Xia} {and} \bibinfo{person}{David~M.
  Pennock}.} \bibinfo{year}{2011}\natexlab{}.
\newblock \showarticletitle{An efficient Monte-Carlo algorithm for pricing
  combinatorial prediction markets for tournaments}. In
  \bibinfo{booktitle}{\emph{Proceedings of the 22nd International Joint
  Conference on Artificial Intelligence}}. \bibinfo{pages}{452--457}.
\newblock


\end{thebibliography}

\onecolumn
\newpage
\appendix
\section{Deferred Proofs}
\label{app:proofs}
\subsection{Proof of Theorem~\ref{thm:lmsr_price}}
The binary-search property implies that the nodes $z$ included in the price calculation (lines 5 and 9) form the cover of $I$, so the algorithm correctly returns the price of $I$. The running time follows thanks to height balance, which implies the depth of the tree is $\order(\log\nvals)$.

\subsection{Proof of Theorem~\ref{thm:lmsr_buy}}
\label{app:lmsr_buy}
We start by showing that a rotation at node $z$ preserves its \textit{partial normalization correctness}.
There are two kinds of rotations, depicted in \Fig{rotations}.
The \emph{left rotation} takes as input a node $z$, with children denoted $z_1$ and $z_{23}$, and children of $z_{23}$ denoted $z_2$ and $z_3$, and rearranges these relationships by removing the node $z_{23}$ and creating a node $z_{12}$, such that $z$ now has children $z_{12}$ and $z_3$, and $z_{12}$ has children $z_1$ and $z_2$.
The \emph{right rotation} is the symmetric operation.

\begin{figure}[h!]
	\centering
	\includegraphics[height=2.2cm]{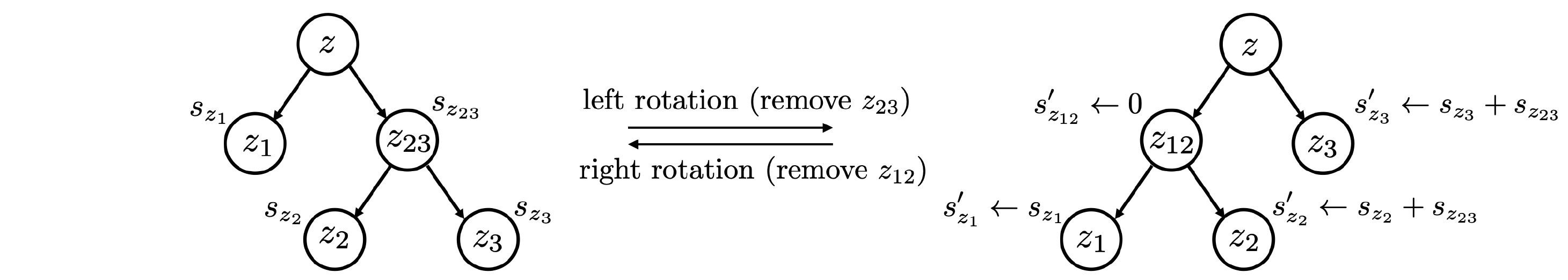}
	\vspace{-1ex}
	\caption{Left and right rotations with node $z$ as an input. Depicted update corresponds to the left rotation.}
	\label{fig:rotations}
	\vspace{-1ex}
\end{figure}

The full procedure of \texttt{RotateLeft} is described in Algorithm~\ref{algo:rotate}. When performing rotations, we need to ensure that the node removal (i.e., removal of $z_{23}$ in left rotation and of $z_{12}$ in right rotation) does not impact the market state.
We achieve this by moving the shares from the removed node into its children, so at the time of removal it holds zero shares (see the right-hand side of \Fig{rotations}, and line 4 of Algorithm~\ref{algo:rotate}).

\begin{algorithm}[h!]
	\caption{Left rotation at node $z$ (right rotation is symmetric).}
	\label{algo:rotate}
	\small
	\begin{algorithmic}[1]
		\algnotext{EndProcedure}
		\State Define subroutines:
		\Statex \algind \Call{ResetInnerNode}{$z$}:
		reset $h_z$ and $S_z$ based on the children of $z$ and the value $s_z$:
		\Statex \algind\algind
		$h_{z}\gets 1+\max\set{h_{\lt{z}},h_{\rt{z}}}$,
		$S_z\gets e^{s_z/b}(S_{\lt{z}}+S_{\rt{z}})$
		\Statex \algind \Call{AddShares}{$z,s$}:
		increase the number of shares held in $z$ by $s$:
		\Statex \algind\algind
		$s_z\gets s_z+s$, $S_z\gets e^{s/b} S_z$
		\medskip
		\Procedure {RotateLeft}{$z$}:
		\smallskip
		\State Let $z_1=\lt{z}$, $z_{23}=\rt{z}$, $z_2=\lt{z_{23}}$, $z_3=\rt{z_{23}}$
		\State \Call {AddShares}{$z_2, s_{z_{23}}$},
		\Call{AddShares}{$z_3, s_{z_{23}}$},
		delete node $z_{23}$
		\State Let $z_{12}$ be a new node with:
		\Statex \algind\algind
		$\lt{z_{12}}=z_1$, $\rt{z_{12}}=z_2$, $I_{z_{12}}=I_{z_1}\cup I_{z_2}$, $s_{z_{12}}=0$
		\State \Call{ResetInnerNode}{$z_{12}$}
		\State Update node $z$:
		\Statex \algind\algind
		$\lt{z}\gets z_{12}$, $\rt{z}\gets z_3$, \Call{ResetInnerNode}{$z$}
		\EndProcedure
	\end{algorithmic}
\end{algorithm}

\begin{lemma}
	A rotation operation preserves partial-normalization correctness.
	\label{lem:rotation}
\end{lemma}
\begin{proof}
	We prove that the original partial normalization value of node $z$, $S_z$, is the same as the updated value, $S'_z$, after a left rotation. A right rotation follows symetrically.
	\begin{align*}
	S_z &= e^{s_z/b} \cdot \Parens{S_{z_1}+S_{z_{23}}}\\
	&= e^{s_z/b} \cdot \Parens{S_{z_1}+e^{s_{z_{23}}/b} \cdot \Parens{S_{z_2}+S_{z_3}}}\\
	&= e^{s_z/b} \cdot \Parens{S'_{z_1}+S'_{z_2}+S'_{z_3}}\\
	&= e^{s_z/b} \cdot \Parens{e^{s'_{z_{12}}/b} \cdot \Parens{S'_{z_1}+S'_{z_2}}+S'_{z_3}} \tag{since $s'_{z_{12}}=0$}\\
	&= e^{s_z/b} \cdot \Parens{S'_{z_{12}}+S'_{z_{3}}} = S'_z
	\end{align*}
\end{proof}

\begin{proof}[Proof of \Thm{lmsr_buy}]
	The correctness of the $\buy$ operation follows because the shares are added to the nodes that form the cover of $I$ (lines 5 and 12 in Algorithms~\ref{algo:buy}), and the updates up the search path restore the properties of the LMSR tree (lines 13--17 in Algorithms~\ref{algo:buy}).
	The running time follows from height balance, which implies that the length of the search path
	is $\order(\log n)=\order(\log\nvals)$.
\end{proof}

\subsection{Proof of Theorem~\ref{thm:lcmm_arb_free}}
\label{app:lcmm_arb_free}
We first show that the constraints $\A^\top\vmu=\vzero$ imply that all levels $\ell=0,1,\dotsc,K$ in $\vmu$ are mutually coherent.
To do this, it suffices to show that all pairs of consecutive levels $\ell$ and $\ell+1$ are coherent, i.e., $\mu_y=\mu_{\ly}+\mu_{\ry}$ for all $y\in\Z_\ell$ where we let $\ly=\lt{y}$ and $\ry=\rt{y}$.

We proceed by induction, beginning with $\ell=K-1$.
In this base case, the constraint $\va_y^\top\vmu=0$, expressed in \eq{constr}, states that $b_K \mu_y=b_K \mu_{\ly}+b_K\mu_{\ry}$, implying levels $K-1$ and $K$ are coherent.

Now assume that all the levels $k>\ell$ are mutually coherent.
We aim to show that levels $\ell$ and $\ell+1$ are coherent.
Pick any $y\in\Z_\ell$.
Then the constraint $\va_y^\top\vmu=0$, expressed in \eq{constr}, implies that
\begin{align}
\notag
\biggParens{\,\sum_{k>\ell} b_{k}}\,\mu_y
&=
\sum_{k>\ell}
b_k\sum_{z\in\Z_k:\:z\subset y} \mu_z
\\
\notag
&=
\sum_{k>\ell}
b_k\,\biggParens{\,\sum_{z\in\Z_k:\:z\subseteq\ly} \mu_z + \sum_{z\in\Z_k:\:z\subseteq\ry} \mu_z}
\\
\label{eq:ind:coh}
&=
\sum_{k>\ell}
b_k\,\BigParens{\mu_{\ly} + \mu_{\ry}}.
\end{align}
\Eq{ind:coh} follows because $\ly$ and $\ry$ are in level $\ell+1$, which is coherent with all levels $k\ge\ell+1$ by the inductive assumption.
Thus, we obtain that $\mu_y=\mu_{\ly} + \mu_{\ry}$ for all $y\in\Z_\ell$, establishing the coherence between levels $\ell$ and $\ell+1$ and completing the induction.

To finish the proof, we note that the LCMM prices at level $K$ are determined by $C_K$, so they describe a probability distribution over $\Omega$.
Since $\A^\top\vp(\vtheta)=\vzero$, all
the levels in $\vp(\vtheta)$ are coherent with level $K$,
which means that they correspond to the expectation of $\vphi$ under the probability distribution
described by the prices at level $K$. Thus, $\vp(\vtheta)$ is a coherent price vector and the multi-resolution LCMM
is therefore arbitrage-free.

\subsection{Proof of Theorem~\ref{thm:constant_loss}}
The worst-case loss of an LCMM is bounded by the sum of the worst-case losses of the component markets $C_k$~\cite{DudikLaPe12}. In our case, these are LMSR submarkets with losses bounded by $b_k\log\,\card{\Z_k}$, so the worst-case loss of the resulting LCMM is at most
\[
\sum_{k=1}^{K} b_k \log(2^k) = \sum_{k=1}^{K} b_k (k\log 2) \le B^*\log 2,
\]
proving the theorem.

\subsection{Proof of Theorem~\ref{thm:lcmm_price}}
Algorithm~\ref{algo:price:lcmm} returns the correct price of $I$, because prices are coherent among submarkets and the nodes included in price calculations form a cover of $I$. The running time is proportional to the length of the search path, which terminates, at the latest, once the first node $z$ with $\alpha_z=\alpha$ is reached. The level of this node coincides with the precision of $\alpha$.

\subsection{Proof of Theorem~\ref{thm:lcmm_buy} and Additional Deferred Material from \Sec{lcmm_buy_cost}}
\label{app:arb}


We begin by deriving an identity that will be useful in the following analysis. For this derivation,
let $C$ be an LMSR with the liquidity parameter $b$, defined over an outcome space $\Omega$.
We will derive a relationship between the price vector in a state $\vtheta$ and the price vector in a new state $\vtheta'=\vtheta+\vdelta$, where $\vdelta$ is any bundle restricted to securities in $E$, i.e., $\delta_\omega=0$ for $\omega\not\in E$.
Denoting $\vmu=\vp(\vtheta)$, $\mu_E=p_E(\vtheta)$, and $\vmu'=\vp(\vtheta')$, we have
\begin{align}
\notag
\mu'_\omega
&=
\frac{e^{\theta_\omega/b}e^{\delta_\omega/b}}
{\sum_{\nu\not\in E} e^{\theta_\nu/b} + \sum_{\nu\in E}e^{\theta_\nu/b}e^{\delta_\nu/b}}
\\
\label{eq:price:update}
&=
\frac{\mu_{\omega}e^{\delta_\omega/b}}{1-\mu_E + \sum_{\nu\in E} \mu_\nu e^{\delta_\nu/b}},
\end{align}
where \Eq{price:update} follows by dividing the numerator as well as denominator by $\sum_{\nu\in\Omega}e^{\theta_\nu/b}$.

We next establish correctness of the arbitrage removal
procedure from Algorithm~\ref{algo:buy:lcmm}. The following lemma provides a critical step:
\begin{lemma}
	Fix a level $\ell<K$. Let $\tvtheta$ be a market state in $\Cx$ such that the associated prices, $\vmu=\vpx(\tvtheta)$, are coherent among all levels $k>\ell$. Then, for any $t\in\R$ and any node $y$ with $\level(y)\le\ell$, the prices after buying $t$ shares of $\va_y$, i.e., $\vmu'=\vpx(\tvtheta+t\va_y)$, remain coherent among all levels $k>\ell$.
	\label{lem:ay}
\end{lemma}

To use \Lem{ay} for arbitrage removal, we start with a market state $\tvtheta$ where all levels are coherent.
When a trader buys some shares of a security $\phi_y$, the level $\ell=\level(y)$ loses coherence with other levels.
By buying a certain number of shares of $\va_y$, it is possible to restore coherence between $\ell$ and $\ell+1$, and \Lem{ay} then implies that coherence with all further levels $k>\ell+1$ is also restored. The process of restoring coherence now continues with the parent of $y$ and the bundle $\va_{\pt{y}}$ as implemented in Algorithm~\ref{algo:buy:lcmm}.

\begin{proof}
	Consider two arbitrary levels $k$ and $m$ with $\ell < k < m$.
	Since prices are coherent between levels $k$ and $m$ before buying $t$ shares of $\va_y$, we have, for any $z\in\Z_{k}$,
	\begin{equation}
	\label{eq:ay:coh}
	\mu_z = \sum_{u\in\Z_{m}:\:u\subset z} \mu_u.
	\end{equation}
	Let $\pi_y$ denote the price of $\phi_y$ according to the securities in $\Z_{k}$ and $\Z_{m}$, that is, $\pi_y=\sum_{z\in\Z_{k}:\:z\subset y} \mu_z=\sum_{u\in\Z_{m}:\:u\subset y} \mu_u$. Note that $\pi_y$ might differ from $\mu_y$, because level $\ell$ is not necessarily coherent with levels $k$ and $m$. Let $\smash{\tvtheta'=\tvtheta+t\va_y}$.
	From the definition of matrix $\A$, the updated $\smash{\ttheta'_z}$ and $\smash{\ttheta'_u}$ for any $z \in \Z_{k}$ and $u \in \Z_{m}$ are
	\[
	\ttheta'_z=
	\begin{cases}
	\ttheta_{z}-tb_{k}
	&\text{if $z \subset y$,}
	\\
	\ttheta_{z}
	&\text{otherwise,}
	\end{cases}
	\qquad\qquad
	\ttheta'_u=
	\begin{cases}
	\ttheta_u-tb_{m}
	&\text{if $u \subset y$,}
	\\
	\ttheta_u
	&\text{otherwise.}
	\end{cases}
	\]
	We calculate the new price $\mu'_z$ of any node $z\in\Z_{k}$ and show it equals to the price derived from its descendants $u \in \Z_{m}$. First, if $z\subset y$, then by \eq{price:update} and \eq{ay:coh},
	\begin{align*}
	\mu'_z &=
	\frac{\mu_z e^{-t}}{\pi_y e^{-t}+1-\pi_y}
	= \frac{\sum_{u\in\Z_{m}:\:u\subset z} \mu_u e^{-t}}{\pi_y e^{-t}+1-\pi_y}
	= \sum_{u\in\Z_{m}:\:u\subset z} \mu'_u.
	\intertext{%
		If $z\not\subset y$, then we similarly have}
	\mu'_z &=
	\frac{\mu_z}{\pi_y e^{-t}+1-\pi_y}
	=\frac{\sum_{u\in\Z_{m}:\:u\subset z} \mu_u}{\pi_y e^{-t}+1-\pi_y}
	= \sum_{u\in\Z_{m}:\:u\subset z} \mu'_u.
	\end{align*}
	Thus, prices remain coherent among all levels $m>k>\ell$.
\end{proof}

Building upon \Lem{ay}, the following lemma provides the precise trade required to restore coherence after an update.

\begin{lemma}
	Fix a level $\ell<K$ and a node $y\in\Z_\ell$ and let $\ly=\lt{y}$ and $\ry=\rt{y}$. Let $\tvthetaz$ and $\tvtheta=\tvthetaz+\vdelta$ be market states in $\Cx$, with associated prices $\vmu^0=\vpx(\tvthetaz)$ and $\vmu=\vpx(\tvtheta)$ such that:
	\begin{itemize}
		\item prices $\smash{\vmu^0}$ are coherent among all levels $k\ge\ell$;
		\item $\vdelta$ is a vector that is zero outside descendants of $y$, i.e., $\delta_z=0$ whenever $z\not\subseteq y$;
		\item prices $\vmu$ are coherent among all levels $k>\ell$.
	\end{itemize}
	Let $\smash{\tvtheta'=\tvtheta+t\va_y}$ where
	\[
	t= \frac{b_\ell}{B_{\ell-1}}\log\Parens{\frac{1-\mu_y}{\mu_y}\cdot\frac{\mu_{\ly}+\mu_{\ry}}{1-\mu_{\ly}-\mu_{\ry}}}
	.
	\]
	Then the associated prices $\vmu'=\vpx(\tvtheta')$ are coherent among all levels $k\ge\ell$.
	\label{lem:arb}
\end{lemma}


\begin{proof}
	By \Lem{ay}, adding $t\va_y$ to $\tvtheta$ maintains coherence among levels $k>\ell$, so it
	suffices to show that levels $\ell$ and $\ell+1$ are mutually coherent in $\vmu'$. Thus, we have to show that $\mu'_z=\mu'_{\lt{z}}+\mu'_{\rt{z}}$ for all $z\in\Z_\ell$.
	
	First note that by the assumption on $\vdelta$ and the definition of $\va_y$, we have
	\begin{align*}
	\ttheta^0_z=\ttheta_z=\ttheta'_z
	&\quad\text{for all $z\in\Z_\ell\wo\set{y}$}
	\\
	\ttheta^0_u=\ttheta_u=\ttheta'_u
	&\quad\text{for all $u\in\Z_{\ell+1}\wo\set{\ly,\ry}$.}
	\end{align*}
	Therefore, by \eq{price:update}, we have for all $z\in\Z_\ell\wo\set{y}$
	\begin{equation}
	\label{eq:mu':mu0}
	\frac{\mu'_z}{1-\mu'_y}=\frac{\mu^0_z}{1-\mu^0_y},
	\qquad
	\text{and}
	\qquad
	\frac{\mu'_{\lt{z}}+\mu'_{\rt{z}}}{1-\mu'_{\ly}-\mu'_{\ry}}
	=\frac{\mu^0_{\lt{z}}+\mu^0_{\rt{z}}}{1-\mu^0_{\ly}-\mu^0_{\ry}}.
	\end{equation}
	Since the vector $\vmu^0$ satisfies $\mu^0_z=\mu^0_{\lt{z}}+\mu^0_{\rt{z}}$ for all $z\in\Z_\ell\wo\set{y}$, \eq{mu':mu0}
	implies that we also have
	$\mu'_z=\mu'_{\lt{z}}+\mu'_{\rt{z}}$ for all $z\in\Z_\ell\wo\set{y}$ as long as $\mu'_y=\mu'_{\ly}+\mu'_{\ry}$. Thus, in order to show that levels $\ell$ and $\ell+1$
	are coherent in $\vmu'$, it suffices to show that $\mu'_y=\mu'_{\ly}+\mu'_{\ry}$.
	
	We begin by explicitly calculating $\ttheta'_z$ and $\ttheta'_u$ for any $z\in\Z_\ell$ and any $u\in\Z_{\ell+1}$:
	\[
	\ttheta'_z=
	\begin{cases}
	\ttheta_z+tB_\ell
	&\text{if $z=y$,}
	\\
	\ttheta_z
	&\text{otherwise,}
	\end{cases}
	\qquad\qquad
	\ttheta'_{u}=
	\begin{cases}
	\ttheta_{u}-tb_{\ell+1}
	&\text{if $u\in\set{\ly,\ry}$,}
	\\
	\ttheta_{u}
	&\text{otherwise.}
	\end{cases}
	\]
	Therefore,
	\begin{align*}
	&
	\mu'_y
	=\frac{\mu_y e^{tB_\ell/b_\ell}}{\mu_y e^{tB_\ell/b_\ell}+1-\mu_y}
	=\frac{1}{1+\frac{1-\mu_y}{\mu_y}e^{-tB_\ell/b_\ell}}
	\intertext{and similarly,}
	&
	\mu'_{\ly}+\mu'_{\ry}
	=\frac{(\mu_{\ly}+\mu_{\ry})e^{-t}}{(\mu_{\ly}+\mu_{\ry}) e^{-t}+1-\mu_{\ly}-\mu_{\ry}}
	=\frac{1}{1+\frac{1-\mu_{\ly}-\mu_{\ry}}{\mu_{\ly}+\mu_{\ry}}e^t}.
	\end{align*}
	Thus, it remains to show that
	\[\textstyle
	\frac{1-\mu_y}{\mu_y}e^{-tB_\ell/b_\ell}
	=
	\frac{1-\mu_{\ly}-\mu_{\ry}}{\mu_{\ly}+\mu_{\ry}}e^t,
	\]
	or equivalently:
	\[\textstyle
	\frac{1-\mu_y}{\mu_y}\cdot\frac{\mu_{\ly}+\mu_{\ry}}{1-\mu_{\ly}-\mu_{\ry}}
	=
	e^{t(1+B_\ell/b_\ell)}.
	\]
	But this follows from our choice of $t$ and the fact that $B_{\ell-1}=B_\ell+b_\ell$,
	completing the proof.
\end{proof}

We finish the section with the proof of \Thm{lcmm_buy}.

\begin{proof}[Proof of \Thm{lcmm_buy}]
	Algorithm~\ref{algo:buy:lcmm} correctly updates the tree (and returns the cost), because the shares are added to the nodes that form a cover of $I$, and coherence is then restored by applying \Lem{arb} up the search path.
	Running times of both algorithms are proportional to the length of the search path to the first node $z$ with $\alpha_z=\alpha$, whose level coincides with the precision of $\alpha$.
\end{proof}

\medskip
\medskip
\section{Trading Dynamics and Additional Results}
\label{app:experiments}
\subsection{Trading Dynamics}
\label{app:trading_dynamics}
We simulate a market consisting of ten traders.
The outcome space is $[0,1)$, discretized at the precision $K=10$.
Traders, indexed as $i\in\{1,\dotsc,10\}$, have noisy access to the underlying true signal $p=0.4$.
Trader $i$'s belief takes form of a beta distribution Beta$(a_i, b_i)$ with $a_i \sim$ Binomial$(p, n_i)$, $b_i = n_i-a_i$, and $n_i = 16i$ representing the quality of the agent's observation of the signal $p$.
Each trader $i$ has an exponential utility $u_i(W) = -e^{-W}$, where $W$ is the trader's wealth.
We consider budget-limited cost-based market makers, whose worst-case loss may not exceed a budget constraint $B$.
For LMSR at precision $k$, this means setting the liquidity parameter to $b=B/\log\parens{2^k}$.
In our experiments, we consider two LMSR markets at precision levels 4 and 8, denoted as $\texttt{LMSR}_{k=4}$ and $\texttt{LMSR}_{k=8}$.
On the other hand, a multi-resolution LCMM has an infinite number of choices for its liquidity at each precision level.
To showcase its interpolation ability, we consider LCMM that evenly splits its budget to precision levels 4 and 8, and denote it as $\texttt{LCMM}_{50/50}$.


Each market starts with the uniform prior, i.e., the initial market prices for all outcomes are equal.
In each time step, a uniformly random agent is picked to trade.
The selected agent considers a set of 50 interval securities, with endpoints randomly sampled according to the agent's belief.
The candidate intervals are rounded to the precision of the corresponding market.%
\footnote{As the number of available interval securities grows exponentially as the supported precision increases, we assume agents have a computational limit and can only consider a (sub)set of available securities.}
The agent considers trading the expected-utility-optimizing number of shares for each interval, and ultimately picks the best interval and executes the trade.
The market maker updates prices accordingly, until the market equilibrium is reached (no trader in the market has the incentive to trade).

Following the described protocol, we run markets mediated by the three respective market makers, $\texttt{LMSR}_{k=4}$, $\texttt{LMSR}_{k=8}$, and $\texttt{LCMM}_{50/50}$, over a range of budget constraints.
To decrease variance, we generate 40 controlled simulation traces (described by a sequence of agent arrivals and their draws of the candidate intervals) and run the market makers on those same traces.
Therefore, any change in agent behavior and price convergence is caused by the different cost functions that market makers adopt to aggregate trades.

\subsection{Additional Experiments}
In Section~\ref{sec:exp}, we demonstrated that by splitting the budget between submarkets that offer interval securities at different precisions, the multi-resolution LCMM is able to interpolate the performance of LMSR market makers.
It can aggregate information at the coarser level efficiently, while also achieving accurate belief elicitation at the finer resolution (after sufficiently many trades).
Here we provide numerical results over a wider range of market maker's budget constraints, validating how the multi-resolution LCMM can balance the price convergence behavior of LMSR markets.

Fig.~\ref{fig:U_shape} shows the price convergence error as a function of budget constraint (thus, the liquidity parameter) and the number of trades for the three respective market makers.
Results are averaged over forty random but controlled trading sequences.
The solid lines depict the price convergence error at precision level $k=8$, and the dashed ones for precision level $k=4$.
The minimum point on each curve indicates the optimal budget, or the optimal value of the liquidity parameter to adopt, for the particular cost function and a specific number of trades.

Intuitively, when the budget for running a market is sufficient, a market operator can support interval securities at any fine-grained precision level, or use only a portion of the budget to achieve optimal performance.
However, when the budget for running a market is limited, say B less than 8, the market designer can preferably aggregate information faster at a coarser resolution by limiting the precision of interval endpoints (e.g., adopting $\texttt{LMSR}_{k=4}$).
However, by design, it cannot accurately elicit beliefs at finer resolutions, even when the market is run for a sufficiently long period of time.
The $\texttt{LMSR}_{k=8}$, on the other hand, benefits from a larger number of trades to aggregate more fine-grained information.
Running the two LMSR markets independently may balance this convergence trade-off, but inevitably results in inconsistent prices between the markets.
Given the different convergence properties of separate LMSRs, a multi-resolution LCMM can allocate its budget accordingly to achieve a desired convergence performance, while maintaining coherent prices.
For example, a market designer, who considers information at precision levels $k=4$ and $k=8$ equally important, may divide the budget between the two levels to enjoy faster price convergence at the coarser resolution, while accurately aggregating a full probability distribution of the continuous variable as trading proceeds.

\begin{figure}[t]
	\centering
	\begin{subfigure}[t]{0.32\textwidth}
		\centering
		\includegraphics[width=\textwidth]{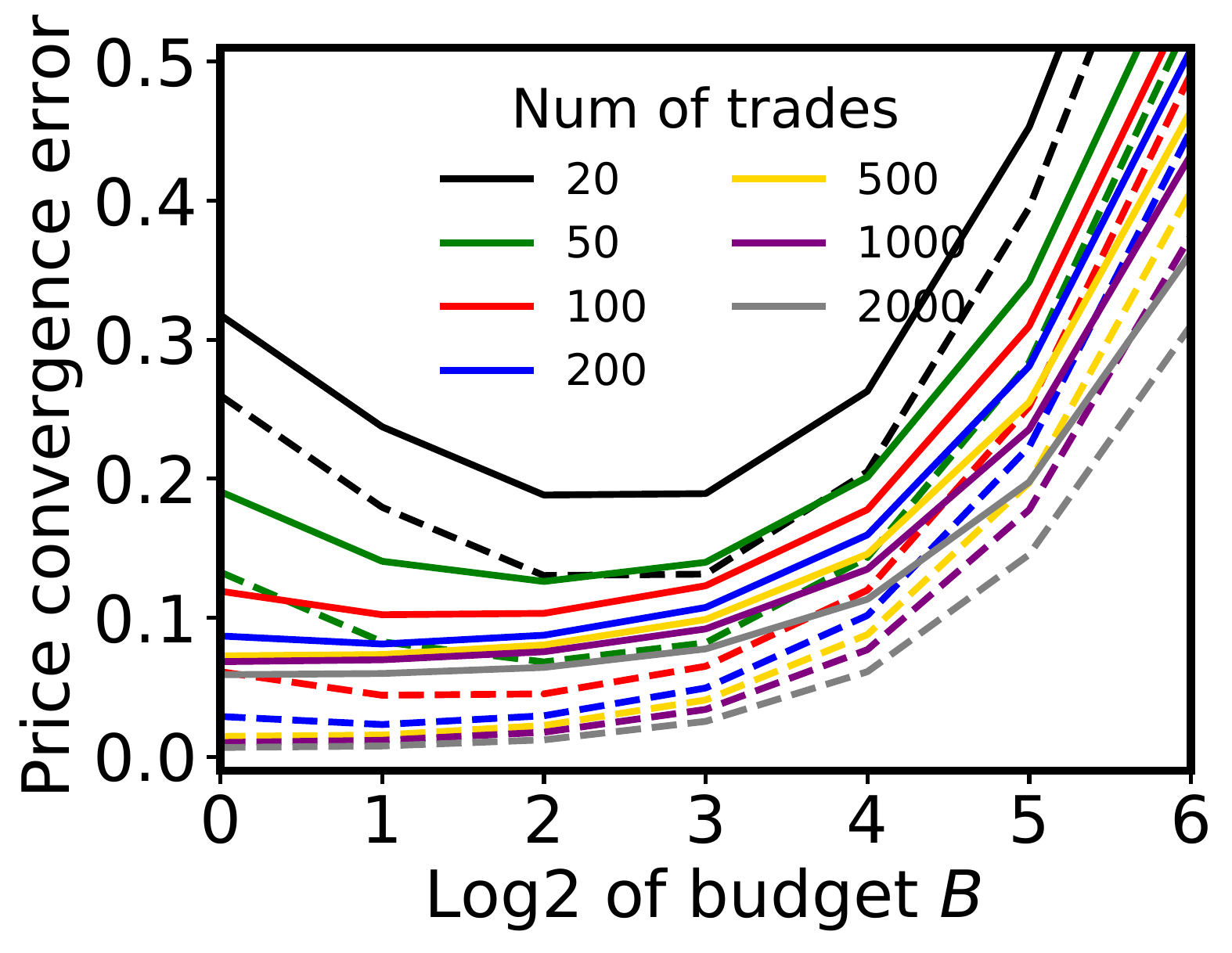}
		\caption{$\texttt{LMSR}_{k=4}$.}
	\end{subfigure}
	\begin{subfigure}[t]{0.32\textwidth}
		\centering
		\includegraphics[width=\textwidth]{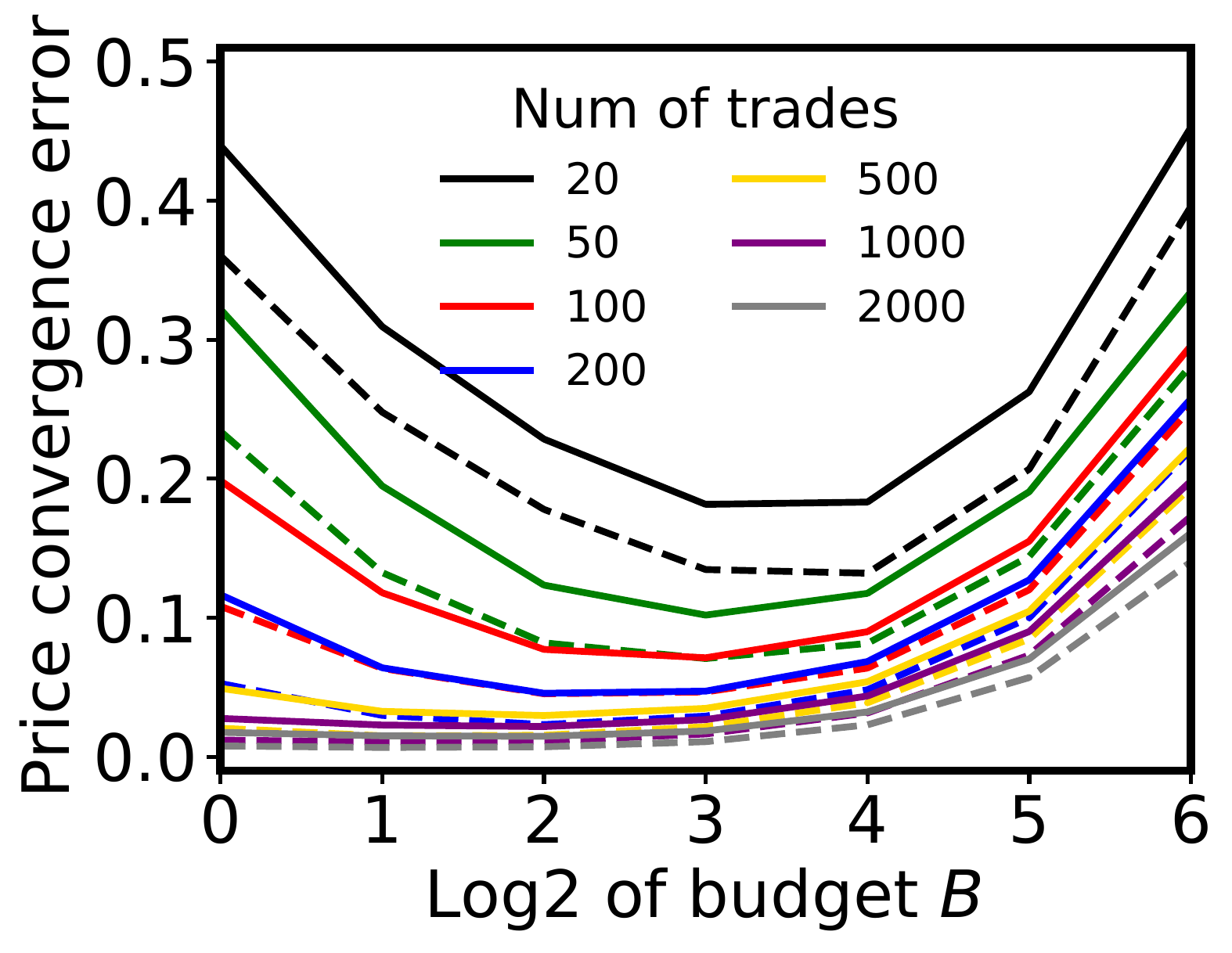}
		\caption{$\texttt{LMSR}_{k=8}$.}
	\end{subfigure}
	\begin{subfigure}[t]{0.32\textwidth}
		\centering
		\includegraphics[width=\textwidth]{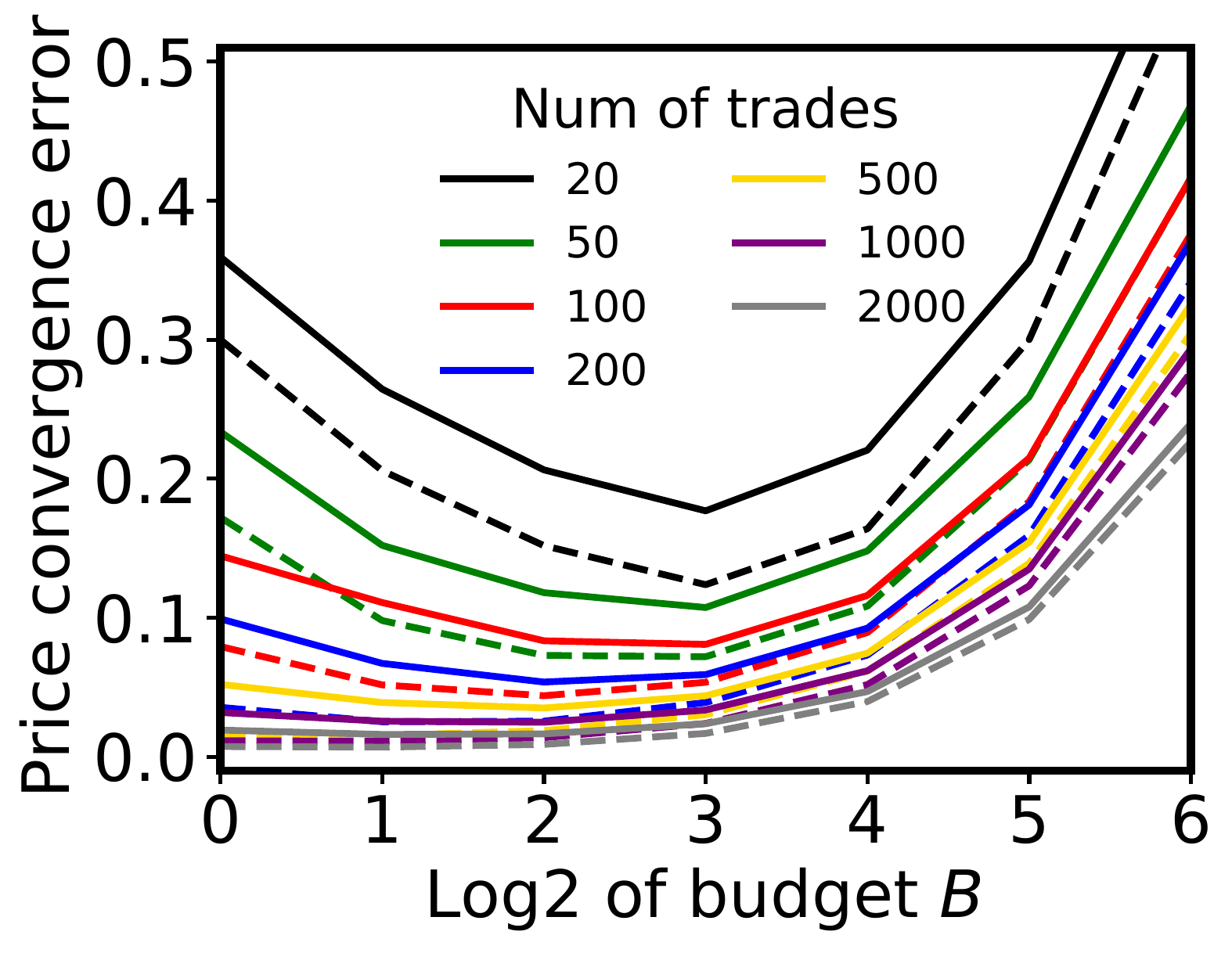}
		\caption{$\texttt{LCMM}_{50/50}$.}
	\end{subfigure}
	\caption{The price convergence error as a function of liquidity and the number of trades (indicated by the color of the line) for the three respective market makers. Solid lines record price convergence error at the finer precision level $k=8$, and dashed ones at the coarser level $k=4$.}
	\label{fig:U_shape}
	\vspace{-2ex}
\end{figure}
\end{document}